\RequirePackage{fix-cm}
\RequirePackage{amsmath}
\documentclass[twocolumn,final]{svjour3}          
\smartqed  
\usepackage{amsmath}
\usepackage{amsfonts}
\usepackage{amssymb}
\usepackage{graphicx}
\usepackage{psfrag}
\usepackage{subfigure}
\usepackage{array}
\usepackage{eurosym}
\usepackage{setspace}
\usepackage{algorithm}
\usepackage{algorithmic}
\usepackage{bm}
\usepackage[numbers,sort&compress]{natbib}
\usepackage{hyperref}
\usepackage{extarrows}
\usepackage{fancynum}
\usepackage{color}
\usepackage{multirow}
\usepackage{makecell}
\usepackage{harpoon}
\usepackage{url}

\newcommand{\myvec}[1]%
{\stackrel{\raisebox{-2pt}[0pt][0pt]{\small$\rightharpoonup$}}{#1}}
\setfnumgsym{\allowbreak}
\journalname{Journal of xxx}
\begin{document}
\title{Segmented Successive Cancellation List Polar Decoding with Tailored CRC
}


\newtheorem{Thm}{Theorem}

\newtheorem{Lem}{Lemma}
\newtheorem{Cor}{Corollary}
\newtheorem{Def}{Definition}
\newtheorem{Exam}{Example}
\newtheorem{Alg}{Algorithm}
\newtheorem{Prob}{Problem}
\newtheorem{Rem}{Remark}
\renewcommand\thesubsection{\thesection.\Alph{subsection}}
\renewcommand{\algorithmicrequire}{\textbf{Input:}}
\renewcommand{\algorithmicensure}{\textbf{Output:}}

\author{Huayi~Zhou \and Xiao~Liang \and Liping~Li \and Zaichen~Zhang \and Xiaohu~You \and Chuan~Zhang$^*$ 
}


\institute{Huayi~Zhou \and Xiao~Liang \and Zaichen~Zhang \and Xiaohu~You \and Chuan~Zhang$^*$ \at
              Lab of Efficient Architectures for Digital-communication and Signal-processing (LEADS),\\
              National Mobile Communications Research Laboratory, \\
              Southeast University, Nanjing, China \\
              \email{\{hyzhou, xiao\_liang, zczhang, xhyu, \\chzhang\}@seu.edu.cn}\\
              $^*$corresponding author
              \and
              Liping~Li\at
              Key Laboratory of Intelligent Computing and Signal Processing of the MoE, Anhui University, Hefei, China\\
              \email{liping\_li@ahu.edu.cn}
}

\date{Received: October 00, 2017 / Accepted: date}

\maketitle

\begin{abstract}
As the first error correction codes provably achieving the symmetric capacity of binary-input discrete memory-less channels (B-DMCs), polar codes have been recently chosen by 3GPP for eMBB control channel. Among existing algorithms, CRC-aided successive cancellation list (CA-SCL) decoding is favorable due to its good performance, where CRC is placed at the end of the decoding and helps to eliminate the invalid candidates before final selection. However, the good performance is obtained with a complexity increase that is linear in list size $L$. In this paper, the tailored CRC-aided SCL (TCA-SCL) decoding is proposed to balance performance and complexity. Analysis on how to choose the proper CRC for a given segment is proposed with the help of \emph{virtual transform} and \emph{virtual length}. For further performance improvement, hybrid automatic repeat request (HARQ) scheme is incorporated. Numerical results have shown that, with the similar complexity as the state-of-the-art, the proposed TCA-SCL and HARQ-TCA-SCL schemes achieve $0.1$ dB and $0.25$ dB performance gain at frame error rate $\textrm{FER}=10^{-2}$, respectively. Finally, an efficient TCA-SCL decoder is implemented with FPGA demonstrating its advantages over CA-SCL decoder.
\end{abstract}

\section{Introduction}\label{sec:intr}
Polar codes, proposed by Ar{\i}kan~\cite{arikan2009rate,Arikan2009Channel}, are considered as a breakthrough of coding theory. It is shown that polar codes can provably achieve the symmetric capacity of binary-input discrete memory-less channels (B-DMCs)~\cite{Arikan2009Channel}. Besides the capacity achieving performance, the asset of polar coding compared to the state-of-the-art (SOA) is its corresponding low-complexity decoding algorithms. Therefore, polar codes have been adopted by 3GPP for eMBB control channels.

Though linear programming (LP) decoder~\cite{Goela2010On}, successive cancellation (SC) decoder, and belief propagation (BP) decoder~\cite{Arikan2008Reed,Hussami2009Performance} have been proposed for polar codes, their performance is not comparable with maximum likelihood (ML) decoder. Thus, the breadth-first SC decoder named SC list (SCL) decoder, was proposed by~\cite{Vardy2011List,Niu2012list}. Cyclic redundancy check (CRC), widely adopted for error detection, has been proved as a simple and effective enabler for further performance improvement with respect to SCL decoder. Numerical results have shown that, CRC-aided SCL (CA-SCL) decoder~\cite{Niu2012CRC} achieves at least no worse performance than the SOA turbo and low-density parity-check (LDPC) decoders~\cite{Arik:2015}. Usually, CRC is placed at the end of decoding to eliminate invalid candidates before final decision. The disadvantages are: 1) Though has better performance than SCL decoder, CA-SCL decoder still suffers from time and space complexity regarding the list size $L$. 2) For intermediate candidates which have already gone wrong, no early elimination could be taken in time until the decoding end is reached, and the computation afterward is in vain.

To address the complexity and redundancy, \cite{Zhou2016Segmented} proposed a segmented CA-SCL (SCA-SCL) decoder. At the same time, \cite{Hashemi2015Partitioned} independently proposed a partitioned CA-SCL (PSCL) decoder, which is similar as the SCA-SCL decoder but with a different partition method. Both decoders divide code bits into segments and insert CRC bits in between, to rule out invalid candidates per segment rather than to wait until the decoding ends. Thus, they can reduce redundancy while keeping comparable performance as CA-SCL decoders. However, existing decoders usually apply the same CRC length to the same number of information or code bits. Though convenient, those straightforward schemes fail to take the code construction into consideration. It is not clear whether the existing uniform partition schemes are optimal and whether better performance can be achieved with the same number of CRC bits.

To our best knowledge, no existing literature has discussed the CRC distribution for SCA-SCL decoding, and its hardware implementation. Analysing the CRC requirement by unequal-length segments and introducing concepts of \emph{virtual transform} and \emph{virtual length}, this paper devotes itself in figuring out a tailored CA-SCL (TCA-SCL) decoding of improved performance and lower complexity than SOA. An HARQ-TCA-SCL decoding is proposed for further performance improvement. Contributions of this paper are: 1) Efficient CRC distribution is proposed for the first time, showing performance advantage over SOA. 2) This paper does not limit itself to specific decoder design, but proposes a formal TCA methodology, which can be readily applied to any existing SCA-SCL decoders. 3) The efficient implementation methodology is also proposed and verified with FPGA implementations.

The remainder of the paper is organized as follows. Section~\ref{sec:preli} reviews the preliminaries. Section~\ref{sec:SCA} analyzes the SCA-SCL decoders for possible refinement. The TCA-SCL decoding is given in Section~\ref{sec:DC-SSCL}. The HARQ-TCA-SCL decoding is given in Section~\ref{sec:arq}. Section~\ref{sec:percom} gives the performance and complexity analysis of the proposed decoding schemes. Section~\ref{sec:arch} proposes a hardware architecture for TCA-SCL decoding. FPGA implementations are given in the same section. Finally, Section~\ref{sec:CON} concludes the entire paper.

\newcounter{mytempeqncnt}
\begin{figure*}[!t]
\normalsize
\setcounter{mytempeqncnt}{\value{equation}}
\setcounter{equation}{5}
\begin{equation}
\label{eqn_dbl_x}
\begin{aligned}
L_N^{(2i - 1)}(y_1^N,\hat u_1^{2i - 2}|{u_{2i - 1}}) &= \max{}^*(L_{N/2}^{(i)}(y_1^{N/2},\hat u_{1,o}^{2i - 2} \oplus \hat u_{1,e}^{2i - 2}|{u_{2i - 1}}) + L_{N/2}^{(i)}(y_{N/2 + 1}^N,\hat u_{1,e}^{2i - 2}|0),\\
& \qquad \qquad L_{N/2}^{(i)}(y_1^{N/2},\hat u_{1,o}^{2i - 2} \oplus \hat u_{1,e}^{2i - 2}|{{\bar u}_{2i - 1}}) + L_{N/2}^{(i)}(y_{N/2 + 1}^N,\hat u_{1,e}^{2i - 2}|1)),\\
L_N^{(2i)}(y_1^N,\hat u_1^{2i - 1}|{u_{2i}}) &= L_{N/2}^{(i)}(y_1^{N/2},\hat u_{1,o}^{2i - 2} \oplus \hat u_{1,e}^{2i - 2}|{u_{2i - 1}} \oplus {u_{2i}}) + L_{N/2}^{(i)}(y_{N/2 + 1}^N,\hat u_{1,e}^{2i - 2}|{u_{2i}}).
\end{aligned}
\end{equation}
\setcounter{equation}{\value{mytempeqncnt}}
\hrulefill
\vspace*{4pt}
\end{figure*}

\section{Preliminaries}\label{sec:preli}
\subsection{Polar Codes}
Denote the input alphabet, output alphabet, and transition probabilities of a B-DMC by $\mathcal{X}$, $\mathcal{Y}$, and $W(y|x)$. With block length $N=2^n$, the information vector, encoded vector, and received vector are $u^{N}_{1}=(u_{1},...,u_{N})$, $x^{N}_{1}=(x_{1},...,x_{N})$, and $y^{N}_{1}=(y_{1},...,y_{N})$. The polar encoding is given by
\begin{equation}\label{encoding}
  x_{1}^{N} = u_{1}^{N}G_{N} = u_{1}^{N}B_{N}F^{\otimes n},
\end{equation}
where $G_{N}$ and $B_{N}$ are the generation matrix and bit-reversal permutation matrix respectively, and $F = \left[\begin{smallmatrix} 1 & 0 \\ 1 & 1 \end{smallmatrix}\right]$. Transmitting channels between $x^{N}_{1}$ and $y^{N}_{1}$ are \\
$W^{(i)}_{N}(y^{N}_{1},u^{i-1}_{1}|u_{i})$, derived by channel combining
\begin{equation}\label{combining}
  W^{N}(y^{N}_{1}|x^{N}_{1})= W^{N}(y^{N}_{1}|u^{N}_{1}G_{N})
\end{equation}
and channel splitting
 \begin{equation}\label{eq:tp}
 \resizebox{.91\hsize}{!}{$
 W^{(i)}_{N}(y^{N}_{1},u^{i-1}_{1}|u_{i}) = \sum _{u^{N}_{i+1}} \frac{1}{2^{N-1}}W^{N}(y^{N}_{1}|x^{N}_{1}), i = 1,...,N.
 $}
\end{equation}

Define $I(W)$ as the symmetric capacity. For B-DMC $W$ and $\delta \in (0,1)$, ${W^{(i)}_{N}}$ polarizes: as $N$ goes to infinity via powers of $2$, $I(W^{(i)}_{N})\in (1-\delta,1]$ approaches $I(W)$ and $I(W^{(i)}_{N})\in [0,\delta)$ approaches $(1-I(W))$. In $(N,K)$ codes, the $K$ most reliable channels with indices in information set $\mathcal{A}$ are chosen to transmit the $K$ information bits in $u_{1}^{N}$; whereas the others, with indices in frozen set $\mathcal{A}_{c}$, transmit the $(N-K)$ frozen bits.
\subsection{SC and SCL Polar Decoders}\label{subsec:scl}
The SC polar decoding tree is a full binary tree. Fig. \ref{fig:SC} shows a toy example for $N = 8$. For each node at the $n$-th level, two possible choices are $0$ and $1$. Each set consisting of all the leaf nodes is associated with a unique estimated codeword ${\hat u}^{N}_{1} = (\hat u_{1}, \hat u_{2}, ..., \hat u_{N})$. If $i \in \mathcal{A}_{c}$, $\hat u_{i} = 0$. Otherwise, the SC decoder computes its log-likelihood ratio (LLR):
   \begin{equation}\label{eq:2}
       L^{(i)}_{N}(y^{N}_{1},\hat u^{i-1}_{1}) = \log {\frac {W^{(i)}_{N}(y^{N}_{1},\hat u ^{i-1}_{1}\mid u_{i} = 0)}{W^{(i)}_{N}(y^{N}_{1},\hat u ^{i-1}_{1}\mid u_{i} = 1)}},
   \end{equation}
and generates its decision as
\begin{equation}\label{eq:de}
{\hat u_i} =
\left\{
{\begin{aligned}
&{0,~{\rm{if}}~L_N^{(i)}(y_1^N,\hat u_1^{i - 1}) \ge 0};\\
&{1,~{\rm{otherwise}}}.
\end{aligned}}
\right.
\end{equation}

The LLR updating is conducted based on the two equations listed in Eq.~(\ref{eqn_dbl_x}). $\max ^*$ denotes the Jacobi logarithm:
\setcounter{equation}{6}
\begin{equation}\label{eq:ja}
\resizebox{.91\hsize}{!}{$
{\max ^*}({x_1},{x_2}) \buildrel \Delta \over = \ln ({e^{{x_1}}} + {e^{{x_2}}})
 = \max({x_1},{x_2}) + \ln(1 + {e^{ - \left| {{x_1} - {x_2}} \right|}}).
$}
\end{equation}

This recursive process starts from each (sub-)tree's root and always traverses the left branch before the right (Fig.~\ref{fig:SC}). When the leaf level is reached, hard decision is made and returned to the parent node.

\begin{figure}[htb]
\centering
\includegraphics[width=0.45\textwidth]{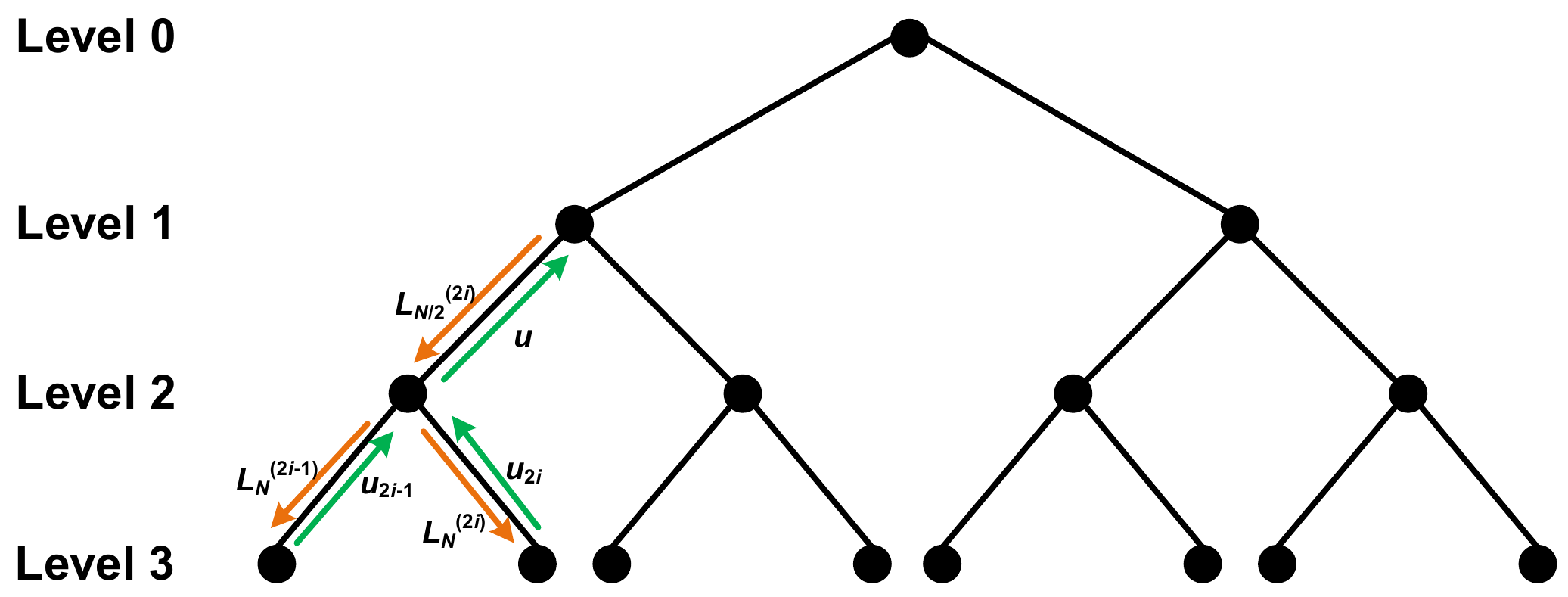}
\caption{Tree illustration of SC decoding process.}\label{fig:SC}
\end{figure}

As a greedy search algorithm, SC decoding keeps only one path based on step-wise decision, with complexity of $\mathcal{O}(N\log N)$. However, this single-candidate method only guarantees the local optimality, and will possibly result in incorrect result. To this end, the SCL decoding, which keeps a list of $L$ survivals, was proposed by~\cite{Vardy2011List,Niu2012list} independently. Fig.~\ref{fig:SCL} illustrates the difference between SC and SCL algorithms. The complexity of SCL decoder is $\mathcal{O}(LN\log N)$. At the $i$-th step, if $i \in \mathcal{A}$, the SCL decoder splits each current path into two paths with both $\hat{u}_{i}=0$ and $\hat{u}_{i}=1$. Out of the $2L$ paths, only the $L$ best ones are kept. Finally, the decoder chooses the best path at the end of decoding process.

    \begin{figure}[htb]
        \centering
            \includegraphics[width = 8 cm]{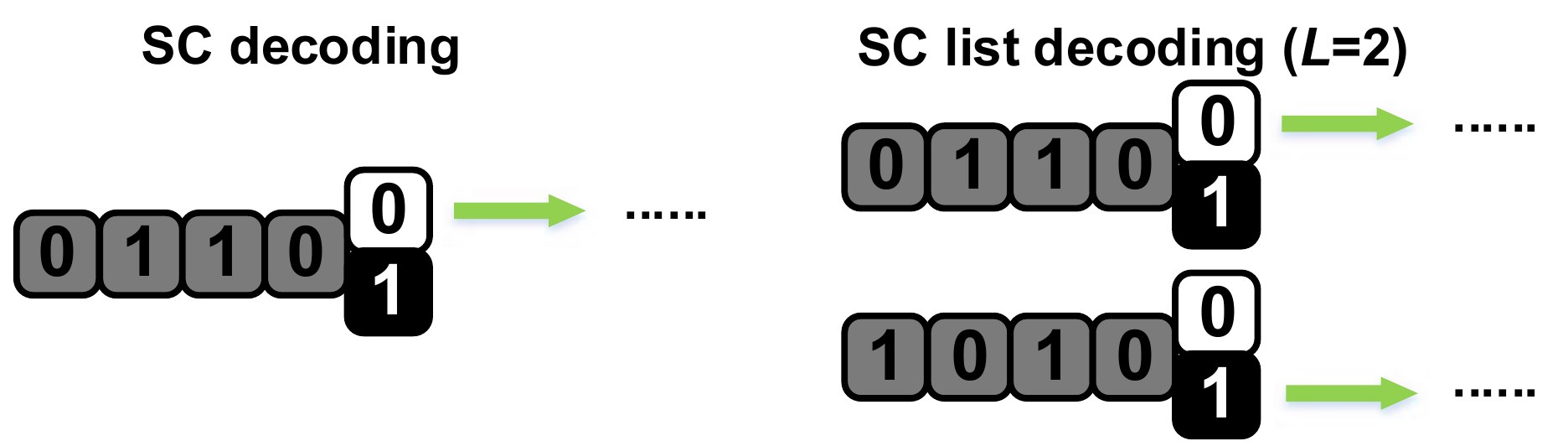}
        \caption{SC decoding and SCL decoding with $L = 2$.}\label{fig:SCL}
    \end{figure}
    \begin{figure}[htb]
        \centering
            \includegraphics[width = 7.5 cm]{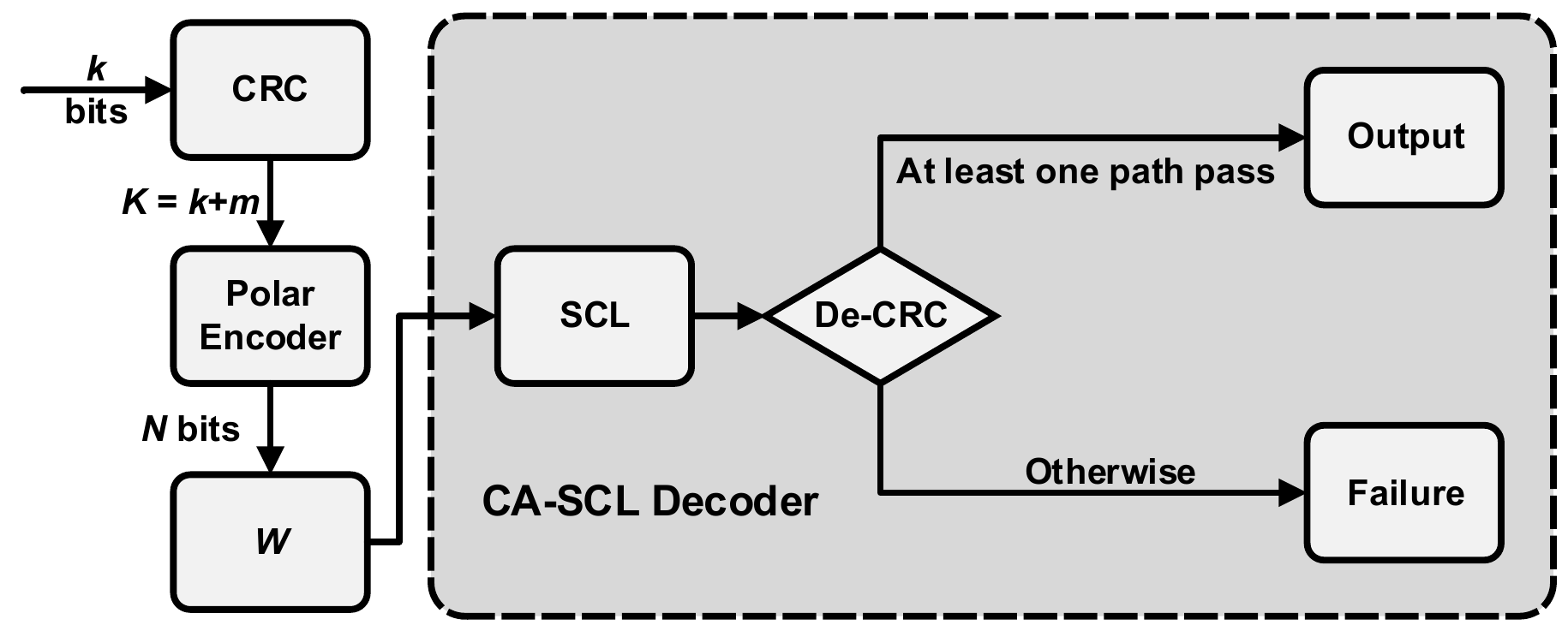}
        \caption{CA-SCL polar decoding.}\label{fig:CASCL}
    \end{figure}
\subsection{CA-SCL Polar Decoder}\label{subsec:cascl}
For further improvement, CA-SCL decoder introduces CRC as a detection tool at the end of decoding~\cite{Niu2012CRC}. Illustrated in Fig.~\ref{fig:CASCL}, CRC detector helps to decide which candidates are possibly correct before metric comparison. Here, $m$ denotes the number of CRC bits. The CRC-passed candidate with the largest metric value is chosen as the final result. If no candidate passes the CRC detection, a decoding failure is claimed.

\section{Segmented CA-SCL Decoding Schemes}\label{sec:SCA}
In this section, we first introduce two SCA-SCL decoding schemes, then propose a refined version. Without loss of generality, the $(1024,512)$ code~\cite{Arikan2009Channel} is employed as a running example, whose polarization is in Fig.~\ref{fig:IW}. Here $W$ is a BEC with erasure probability $\epsilon = 0.5$, $I(W^{(i)}_{N})$ is computed by:
    \begin{equation}\label{eq:3}
    \left\{
       \begin{aligned}
       &I(W^{(2i-1)}_{N})=I(W^{(i)}_{N/2})^{2},\\
       &I(W^{(2i)}_{N})=2I(W^{(i)}_{N/2})-I(W^{(i)}_{N/2})^{2};
       \end{aligned}
       \right.
    \end{equation}
with $I(W^{(1)})=1-\epsilon$. The blue stars in Fig.~\ref{fig:IW} denote the information bits, whereas the red points denote the frozen bits.

   \begin{figure}[htb]
        \centering
            \includegraphics[width = 8 cm]{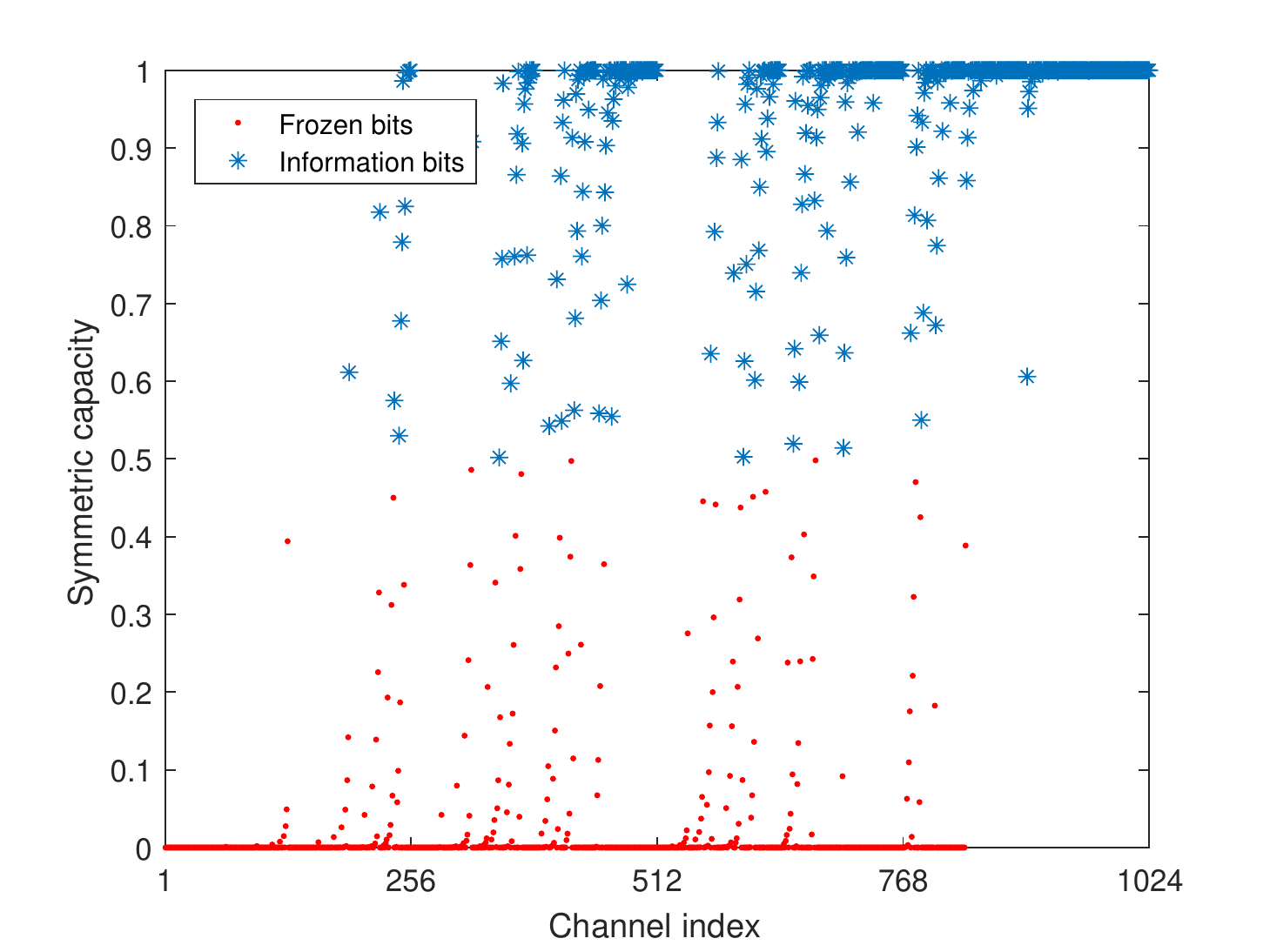}
        \caption{Channel polarization for a BEC with $\epsilon = 0.5.$}\label{fig:IW}
    \end{figure}

\subsection{Comparison of Different Segmented Schemes}\label{subsec:comp}
To the authors' best knowledge, there are two segmented CRC-aided SCL methods. The PSCL scheme proposed in~\cite{Hashemi2015Partitioned} aims to reduce memory consumption, and applies uniform partitions to code bits for implementation convenience. The hardware reduction comes at the cost of some performance loss compared to the conventional CA-SCL algorithm and always forces the number of candidate paths to $1$ after each CRC. The SCA-SCL scheme proposed in~\cite{Zhou2016Segmented} aims to reduce both the time and space complexity. Uniform segments are applied to information bits and CRC is employed as a tool to eliminate decoding redundancy without harming the performance.
    \begin{figure}[htb]
        \centering
            \includegraphics[width = 8 cm]{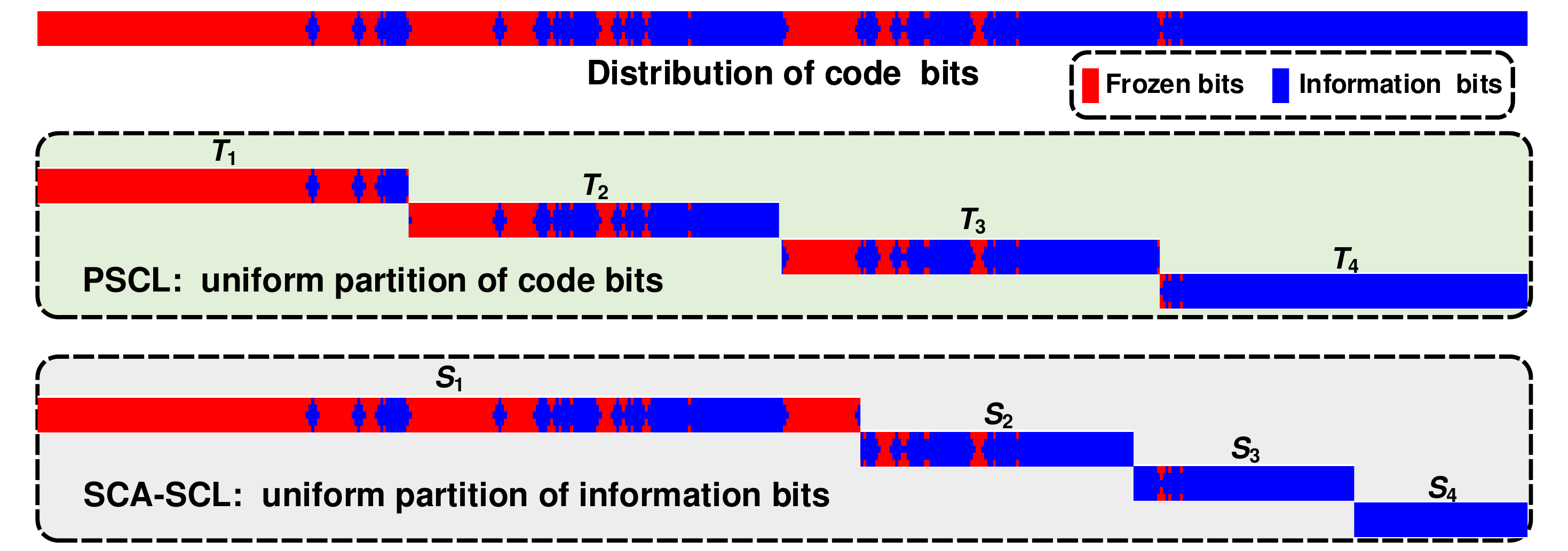}
        \caption{Different segmented decoding schemes with $P = 4$.}\label{fig:seg}
    \end{figure}

Let $P$ denote the number of segments. For PSCL decoder, the index set of Segment-$i$ is $T_i$ ($1 \leq i \leq P$). $|\cdot|$ denotes the cardinality of one set. We have
\begin{equation}\label{eq:pscl}
  \sum\nolimits_{i=1}^{P} {|T_i|}=N,~\textrm{and}~{|T_i|}=N/P.
\end{equation}
For SCA-SCL decoder, the index set of Segment-$i$ is $S_i$ ($1 \leq i \leq P$). we have
\begin{equation}\label{eq:scascl}
  \sum\nolimits_{i=1}^{P} {|S_i|}=N,~\textrm{and}~|{S_i} \cap \mathcal{A}|=K/P.
\end{equation}
One simple example of $P = 4$ is illustrated in Fig.~\ref{fig:seg}.

Theoretically, both schemes are similar but differ in the partition methods. PSCL decoding employs uniform code bit partition, which is implementation friendly. However, since only one candidate can survive after each CRC, small performance degradation is expected, especially in low SNR region. SCA-SCL decoding employs uniform information bit partition, which can keep the performance as CA-SCL decoding while successfully reducing the space and time complexity. This advantage comes from the decoding flexibility. However, the flexibility will make the implementation more complicated.

\subsection{PSCL with Early Termination}\label{subsec:c}
The first observation is that, both schemes apply the same CRC to the uniformly partitioned segments. Without looking into the symmetric capacity of each binary channel, this straightforward scheme may not be optimal. The second observation is that both schemes have their own merits, it would be smarter to merge them together. In other words, it is estimated that we can propose a new approach which is both implementation friendly and adaptive.

One simple mixture of both schemes is to introduce early termination to PSCL decoding. However, this simple combination may not be reasonable in certain cases. Fig.~\ref{fig:PSCL} gives an example with $P=4$. Shown in Fig.~\ref{fig:segc}, for the $(1024,512)$ code with $m=32$, the information lengths of four segments are $20$, $123$, $156$, and $245$, respectively. If uniform CRC bits are employed, the first segment has $|T'_1| = |T_1| - |C_1| = 12$ information bits and the last segment has $|T'_4| = |T_4| - |C_4| = 237$ information bits. It is unreasonable to use the same $8$-bit CRC to both the $12$-bit and $237$-bit segments. To this end, the TCA-SCL decoding is proposed in the following section.
    \begin{figure*}[htb]
        \centering
            \includegraphics[width = 15 cm]{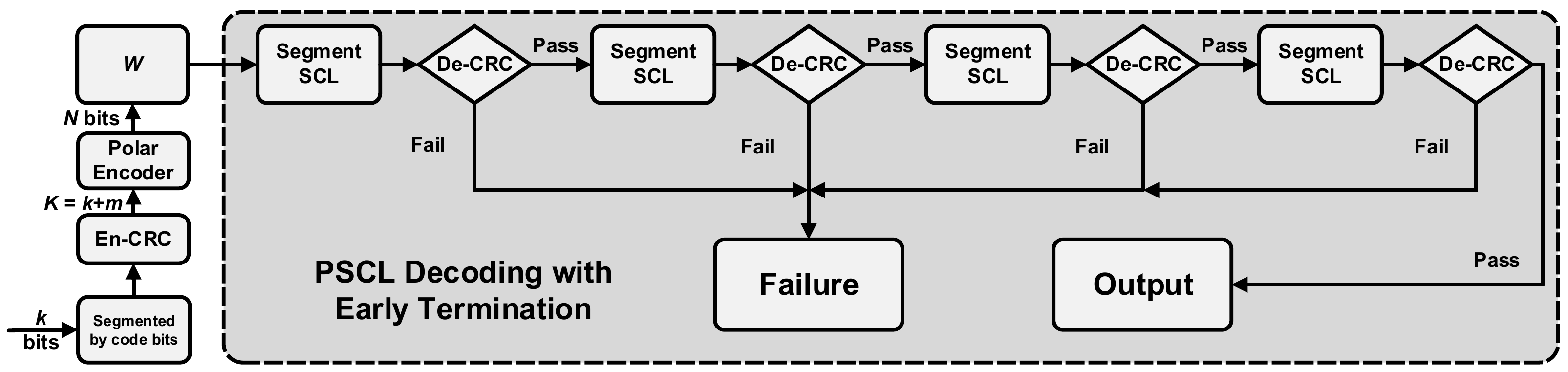}
        \caption{PSCL decoding with early termination.}\label{fig:PSCL}
    \end{figure*}
    \begin{figure*}[htb]
        \centering
            \includegraphics[width = 16 cm]{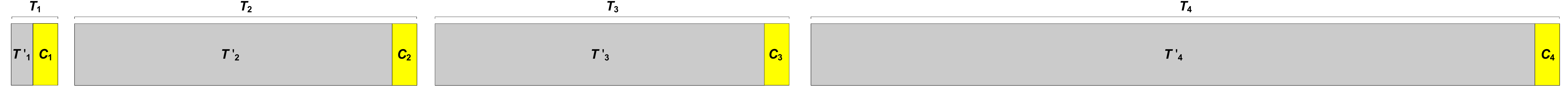}
        \caption{The CRC allocation of PSCL ($N=1024$, $K=512$).}\label{fig:segc}
    \end{figure*}

\section{CA-SCL Decoding with Tailored CRC}\label{sec:DC-SSCL}
In this section, we first discuss how to measure the requirement of CRC bits for different segments. Then the concepts of virtual transform and virtual length are introduced. A visualization method of polarized channel's symmetric capacity is also proposed. The detailed TCA-SCL decoding is finally proposed. It should be noted that though the TCA-SCL decoding is based on uniform partition of code bits, it can be readily applied to other uniform or nonuniform partition schemes.

\subsection{Requirement of CRC Length for Polar Codes}\label{subsec:crc}
Assume a total of $m$ CRC bits are available and are divided into $P$ segments $C_1, C_2, \ldots, C_P$. It may not be suitable to set $|C_1| = |C_2| = \ldots =|C_P|$ for $P$ segments with different lengths. How to measure the requirement of the CRC length for each segment is critical. To the authors' best knowledge, no literature has addressed this specific problem. To maintain the same error detection capability in the situation of independent channels, it is concluded that longer sequence requires more CRC bits~\cite{Koopman2004Cyclic}. However, this conclusion does not suit polar codes because the reliability of different channels are different. A reasonable measurement on requirement of CRC length should take both sequence length and symmetric capacity into account. In the following, concepts of virtual transform and virtual length are proposed to this end.

\subsection{Virtual Transform and Virtual Length}\label{subsec:vt}
Including CRC bits, we always pick the $K+m$ most reliable bits out of $N$ based on the symmetric capacity $I(W^{(i)}_{N})$ with $i \in \mathcal{A'}$. $\mathcal{A'}$ is the new information set including CRC bits, and $|\mathcal{A'}| = K+m$. Calculate $\bar I$ as follows:
\begin{equation}\label{average}
  {\bar I}=\frac{1}{K+m}{\sum\limits_{i \in \mathcal{A'}} {I(W^{(i)}_{N})}}.
\end{equation}

\begin{Def}\label{def:vt}
  \textbf{Virtual Transform} To operate the virtual transform, we first calculate $I'(i)$:
  \begin{equation}\label{eq:v1}
  I'(i) = \bar I/I(W^{(i)}_{N}).
  \end{equation}
  The virtual value of the channel is
\begin{equation}\label{eq:v2}
J(i) =
\left\{
{\begin{aligned}
&{1 +  \frac{(I'(i) - 1)}{2(1-\bar I)}   ,~{\rm{if}} ~ I'(i) \geq 1};\\
&{1 -  \frac{(1 - I'(i))}{2(1-\bar I)}   ,~{\rm{if}} ~ I'(i) < 1}.
\end{aligned}}
\right.
\end{equation}
\end{Def}

\begin{Def}\label{def:vl}
  \textbf{Virtual Length} The summation of $J(i)$ in the $k$-th segment is its virtual length:
  \begin{equation}\label{eq:v3}
  v{l_k} = \sum\limits_{i \in \{{{T_k} \cap {\mathcal{A'}}}\}} {J(i)}.
  \end{equation}
\end{Def}

The CRC allocation is given by
\begin{equation}\label{eq:v4}
|C_1|:\ldots :|C_P|= \textrm{adjust}\left(\frac{m \times {v{l_1}}}{{\sum\nolimits_{i = 1}^P {v{l_i}} }},\ldots,\frac{m \times {v{l_P}}}{{\sum\nolimits_{i = 1}^P {v{l_i}} }}\right),
\end{equation}
where $\textrm{adjust}(\cdot)$ is a function which adjusts the allocation results to near integers and takes the following steps: 1) find an unmarked $k$ which has minimum $|\textrm{ROUND}(\frac{m \times {v{l_k}}}{{\sum\nolimits_{i = 1}^P {v{l_i}} }})-\frac{m \times {v{l_k}}}{{\sum\nolimits_{i = 1}^P {v{l_i}} }}|$, then mark $k$ and set $|C_k|=\textrm{ROUND}(\frac{m \times {v{l_k}}}{{\sum\nolimits_{i = 1}^P {v{l_i}} }})$; 2) repeat step 1) for $(P-2)$ times; 3) Assume the left unmarked index is $k'$. Set $|{C_{k'}}| = m - \sum\limits_{i \ne k'} {|{C_i}|}$, where $1 \le i \le P$.

\subsection{Visualization of Channel Symmetric Capacity}\label{subsec:vis}
\begin{figure*} \centering
\subfigure[Visualization of symmetric capacity for code bits ($N=64$).] { \label{fig:sc:a}
\includegraphics[width = 16.6 cm]{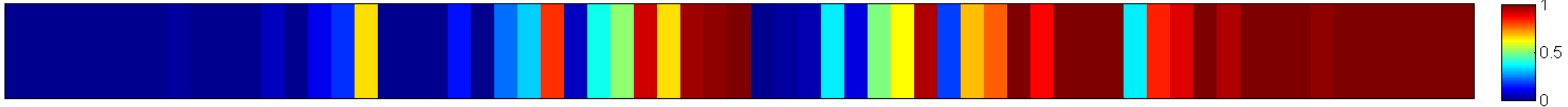}
}
\subfigure[Visualization of symmetric capacity for code bits ($N=1024$).] { \label{fig:sc:b}
\includegraphics[width = 16.55 cm]{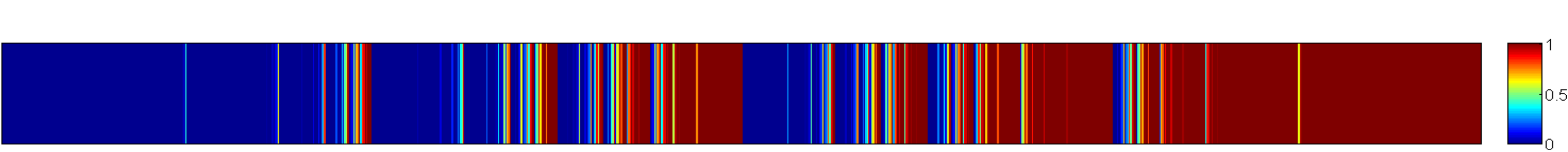}
}
\subfigure[Visualization of symmetric capacity for information bits ($N=1024$, $K=512$).] { \label{fig:sc:c}
\includegraphics[width = 16.55 cm]{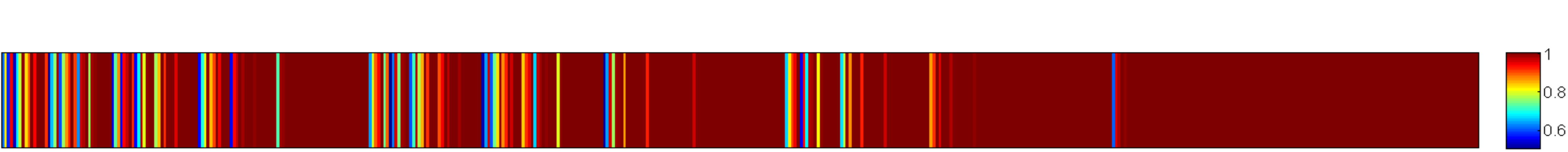}
}
\subfigure[Four segments of information bits with virtual lengths ($N=1024$, $K=512$).] { \label{fig:sc:d}
\includegraphics[width = 16.55 cm]{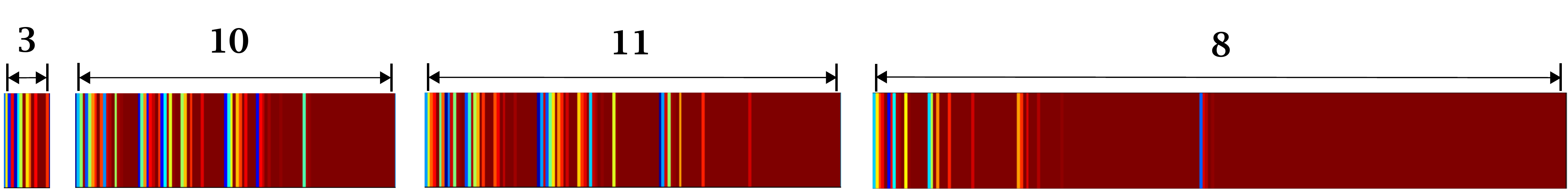}
}
\subfigure[The CRC allocation of TCA-SCL ($N=1024$, $K=512$).] { \label{fig:sc:e}
\includegraphics[width = 16.8 cm]{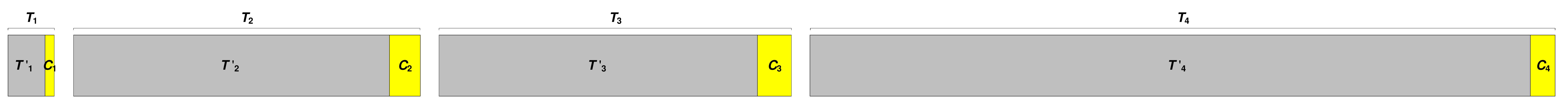}
}
\caption{Visualization illustration of symmetric capacity.}
\label{fig:ex}
\end{figure*}

Before we give more details of the proposed TCA-SCL decoding, one visualization method of symmetric capacity is proposed for easy understanding and illustration. In this visualization, the gradient colors from iridescence are used to demonstrate the symmetric capacity of each channel. According to the legend, the more symmetric capacity approaching $1$ ($0$), the more bathochromic (hypsochromic) it will be. Fig.~\ref{fig:sc:a} shows the visualization for polar codes with $N=64$.

\begin{Exam}\label{exam:1024}
  For $(1024,512)$ polar codes with $32$ CRC bits, visualization of code bits is given in Fig.~\ref{fig:sc:b}. The visualization of $512$ information bits is given in Fig.~\ref{fig:sc:c}. For TCA-SCL decoding, set $P=4$. According to Definition~\ref{def:vl}, the ratio of virtual lengths is:
  \begin{equation}\label{eq:vl:ratio}
    vl_1:vl_2:vl_3:vl_4 = 3.54:9.84:10.91:7.70,
  \end{equation}
  which is illustrated by Fig.~\ref{fig:sc:d}. Then the CRC allocation is obtained according to Eq.~(\ref{eq:v4}):
  \begin{equation}\label{eq:crc:ratio}
    |C_1|:|C_2|:|C_3|:|C_4| = 3:10:11:8.
  \end{equation}
  The refined CRC allocation based on virtual length is given in Fig.~\ref{fig:sc:e}.
\end{Exam}

\begin{Rem}\label{rem:vl}
  Generally speaking, the hypsochromic part in the visualization chart mainly contributes to the virtual length. The more hypsochromic segment requires more CRC bits.
\end{Rem}

This refined SCA-SCL decoding based on virtual length is named TCA-SCL decoding. Details of TCA-SCL decoding is given as follows. The corresponding performance and implementation are discussed in Section~\ref{sec:percom} and Section~\ref{sec:arch}.

\subsection{Tailored CA-SCL Decoding}\label{subsec:tca}
The detailed tailored CA-SCL decoding is given in this subsection. For TCA-SCL encoding, we set $P$ segments and perform the virtual transform to obtain the corresponding virtual lengths. Then we allocate the CRC bits according to the ratio of virtual lengths before polar encoding.

\begin{algorithm} 
\caption{TCA-SCL Polar Encoding} 
\label{alg:1} 
\begin{algorithmic}[1]
\REQUIRE $u^{N}_{1}$, $I(W^{(i)}_{N})$, $N$, $K$, $m$, $P$.
\STATE Set $P$ segments;
\STATE ${\bar I}=\frac{1}{K+m}{\sum\nolimits_{i \in \mathcal{A'}} {I(W^{(i)}_{N})}}$;
\FOR{$i=1$; $i<=N$; $i++$}
\STATE $I'(i) = \bar I/I(W^{(i)}_{N})$;
\IF {$I'(i) \geq 1$}
\STATE $J(i)=1 + \frac{(I'(i) - 1)}{2(1-\bar I)}$;
\ELSE
\STATE $J(i)=1 - \frac{(1 - I'(i))}{2(1-\bar I)}$;
\ENDIF
\ENDFOR
\FOR{$k=1$; $k<=P$; $k++$ }
\STATE $v{l_k} = \sum\nolimits_{i \in \{{{T_k} \cap {\mathcal{A'}}}\}} {J(i)}$;
\ENDFOR
\STATE $\textrm{addCRC}(u^{N}_{1},vl_1,vl_2,..,vl_P)$;
\STATE $x^N_1= \textrm{encoder}(u^N_1)$.
\ENSURE $x^{N}_{1}$.
\end{algorithmic}
\end{algorithm}

Here, $\textrm{addCRC}(\cdot)$ is function which performs Eq.~(\ref{eq:v4}). Function $\textrm{encoder}(\cdot)$ performs conventional polar encoding. For TCA-SCL decoding, SCL decoding with early termination is performed as follows. Here, the function $\textrm{SCL}'(\cdot)$ is the SCL decoding for Segment-$j$. Define $\mathcal{U}_i$ as the the output paths set of SCL($\cdot$) in $i$-th segment. Define $\textrm{passCRC}(\cdot)$ as the function which checks if at least one path of $\mathcal{U}_i$ can pass the CRC. If one or more than one path can pass the CRC, the path with the largest metric of them is chosen to refresh $\hat u^{N}_{1}$.

\begin{algorithm} 
\caption{TCA-SCL Polar Decoding} 
\label{alg:2} 
\begin{algorithmic}[1]
\REQUIRE $y^{N}_{1}$, $N$, $P$, $L$.
\FOR{$i=1$; $i<=P$; $i++$ }
\STATE $\mathcal{U}_i = \textrm{SCL}'(y^{N}_{1},i,L$);
\IF {$\textrm{passCRC}(\mathcal{U}_i)=\textrm{false}$}
\STATE \textbf{break};
\ENDIF
\STATE refresh $\hat u^{N}_{1}$ by the survival path in $\mathcal{U}_i$;
\ENDFOR
\ENSURE $\hat u^{N}_{1}$.
\end{algorithmic}
\end{algorithm}

\section{TCA-SCL Decoding with HARQ}\label{sec:arq}
Besides early termination, the proposed TCA-SCL decoding can also work in a HARQ way when segmented CRC fails. HARQ has been widely used in delay insensitive communication systems for a capacity-approaching throughput~\cite{hagenauer1988rate,rowitch2000performance,yue2007design}. Recently, HARQ has been considered for polar decoding. \cite{chen2013hybrid} introduced a HARQ scheme based on a class of rate-compatible polar codes constructed by performing punctures and repetitions using punctured polar coding~\cite{niu2013beyond}. An incremental redundancy HARQ (IR-HARQ) scheme via puncturing and extending of polar codes is proposed in \cite{saberincremental}. Both algorithms use punctured patterns to suit different rates. However, puncturing causes a performance loss and needs hybrid decoding schemes to remedy it with high complexity. And IR-HARQ scheme needs to retransmit frozen bits one by one after transmitting $K$ information bits. Therefore, the decoding complexity of IR-HARQ is $O(N^{2}\log N)$, which is high for a large $N$.

To overcome this issue, we give a HARQ-TCA-SCL scheme based on TCA-SCL decoding. When a segment decoding failure occurs, the system resends the specific segment and merges the new information bits with the old ones by maximum ratio combining (MRC). For different segments sharing the same SNR, decoder can apply linear superposition to obtain the average value. As the number of segment retransmission goes up, the noise power converge to zero, which helps to improve the performance effectively.
    \begin{figure}[htb]
        \centering
            \includegraphics[width = 8cm]{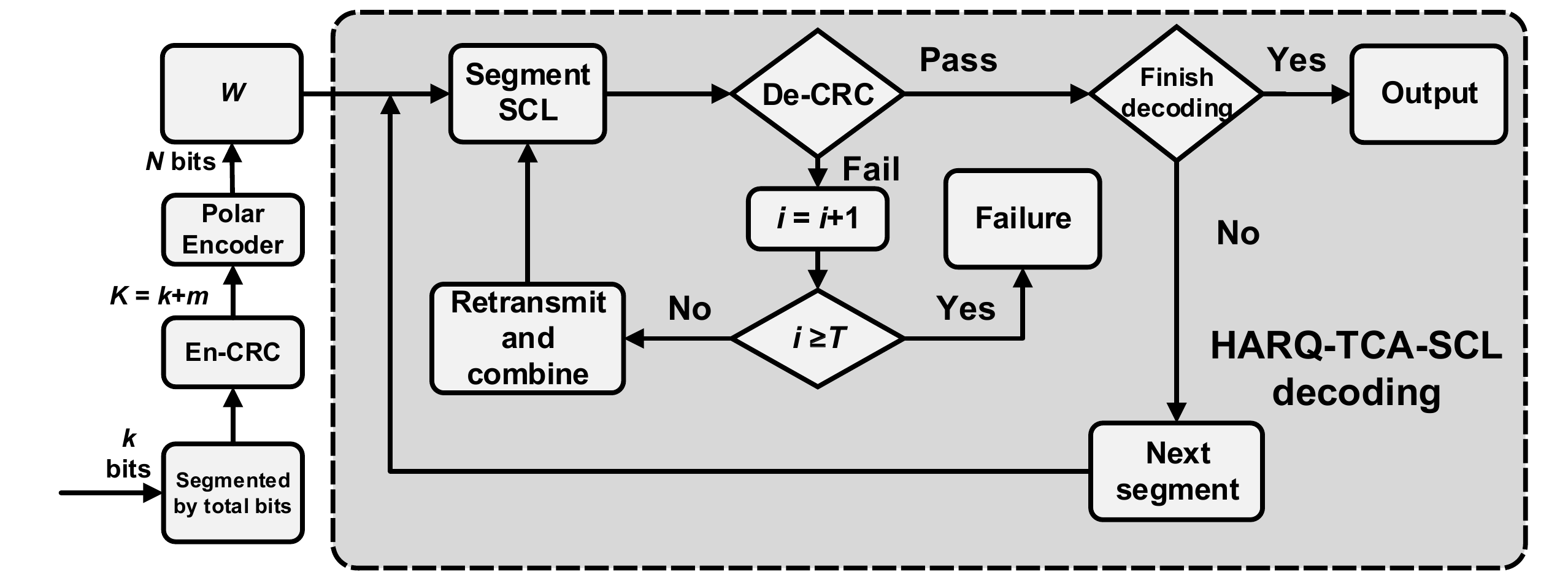}
        \caption{Proposed HARQ-TCA-SCL decoding scheme.}\label{fig:arq}
    \end{figure}

The proposed HARQ-TCA-SCL scheme is illustrated in Fig.~\ref{fig:arq}. Let $i$ denotes the current number of times a transmission attempted, $T$ denotes the maximum retransmission times, and $j$ ($\le P$) denotes the current position of the segments. The details of HARQ-TCA-SCL scheme are listed as follows:
\begin{algorithm} 
\caption{HARQ-TCA-SCL Polar Decoding} 
\label{alg1} 
\begin{algorithmic}[1]
\REQUIRE $y^{N}_{1}$, $T$, $P$, $L$
\STATE $i = 1$;
\FOR{$j=1$ to $P$}
\STATE $\textrm{mark} = \textrm{false}$;
\WHILE {$i < T$ and $\textrm{mark} = \textrm{false}$}
\STATE $\mathcal{U}_j =\textrm{SCL}'(y^{N}_{1},j,L)$;
\STATE $\textrm{mark} = \textrm{passCRC}(\mathcal{U}_j)$;
\IF {$\textrm{mark} = \textrm{false}$}
\STATE $i = i+1$;
\STATE Retransmit and combine Segment-$j$;
\ENDIF
\ENDWHILE
\IF {$\textrm{mark} = \textrm{false}$}
\STATE \textbf{break};
\ENDIF
\STATE refresh $\hat u^{N}_{1}$ by the survival path in $\mathcal{U}_j$;
\ENDFOR
\ENSURE $\hat u_{1}^{N}$
\end{algorithmic}
\end{algorithm}

We initialize $i = 1$ for the HARQ-TCA-SCL decoding, then perform the SCL decoding for Segment-$j$ (function $\textrm{SCL}'(\cdot)$) and obtain CRC results on each survival path at the end of segment SCL decoding. If at least one path can pass CRC, we save the path with the highest probability and move to the next segment. Otherwise, we update $i$ to $i+1$, combine Segment-$j$ with the retransmitted part and the old ones, redo the TCA-SCL decoding. Algorithm terminates with a decoding failure if $i = T$.

\section{Performance and Complexity Analysis}\label{sec:percom}
\subsection{Performance Analysis}\label{subsec:performance}
In this subsection, performance comparison between different algorithms is given with binary-input additive white \emph{Gaussian} noise channels (BI-AWGNCs). Different code lengths, rates, and partition schemes are considered: for Fig.~\ref{fig:fer64}, we have $N = 64$, $K = 36$, $m = 8$, and $P = 2$; for Fig.~\ref{fig:fer1024}, we have $N = 1024$, $K = 512$, $m = 32$, and $P = 4$. The information set $A$ is selected according to~\cite{Arikan2009Channel,tal2013construct}. We use corresponding hex value to represent CRC polynomial. For example, a CRC-$4$ detector with polynomial $g(D) = D^{4}+D+1$ is described as CRC-$4$ ($0$x$9$) in this paper (the `$+1$' is implicit in the hex value). For $(64,36)$ code, we set $2$ copies of CRC-$4$ ($0$x$9$) for (HARQ-)PSCL scheme, and CRC-$5$ ($0$x$12$) and CRC-$3$ ($0$x$5$) for (HARQ-)TCA-SCL scheme. For $(1024,512)$ code, we set $4$ copies of CRC-$8$ ($0$xA$6$) for (HARQ-)PSCL scheme, and CRC-$3$ ($0$x$5$), CRC-$10$ ($0$x$327$), CRC-$11$ ($0$x$583$), and CRC-$8$ ($0$xA$6$) for (HARQ-)TCA-SCL scheme. All the CRC detectors are with the best CRC generation polynomial suggested by~\cite{Koopman2004Cyclic}.

\begin{figure} \centering
\subfigure[($N=64$, $K=36$, $m=8$, $P=2$)] { \label{fig:fer64}
\includegraphics[width = 8 cm]{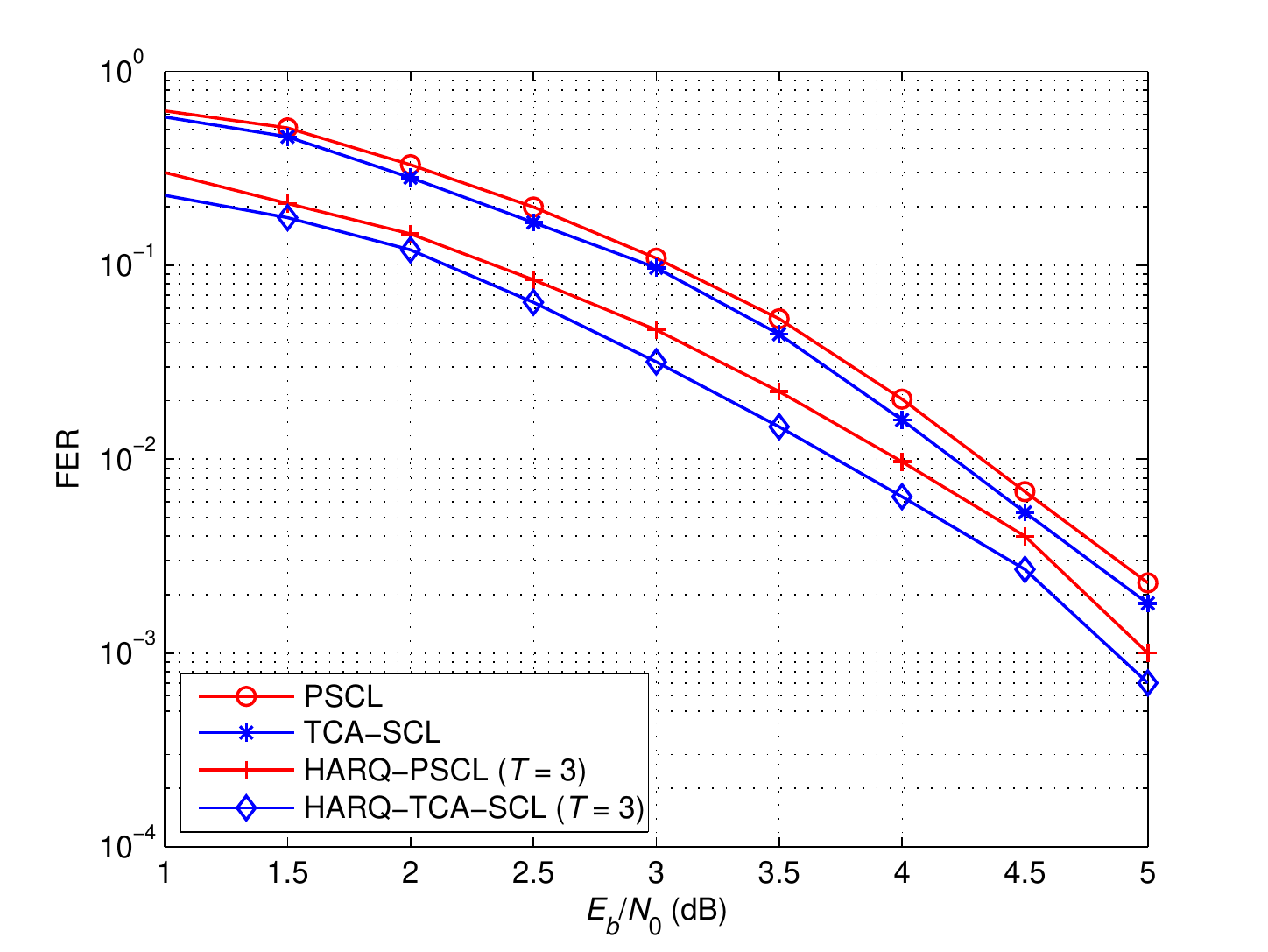}
}
\subfigure[($N=1024$, $K=512$, $m=32$, $P=4$)] { \label{fig:fer1024}
\includegraphics[width = 8 cm]{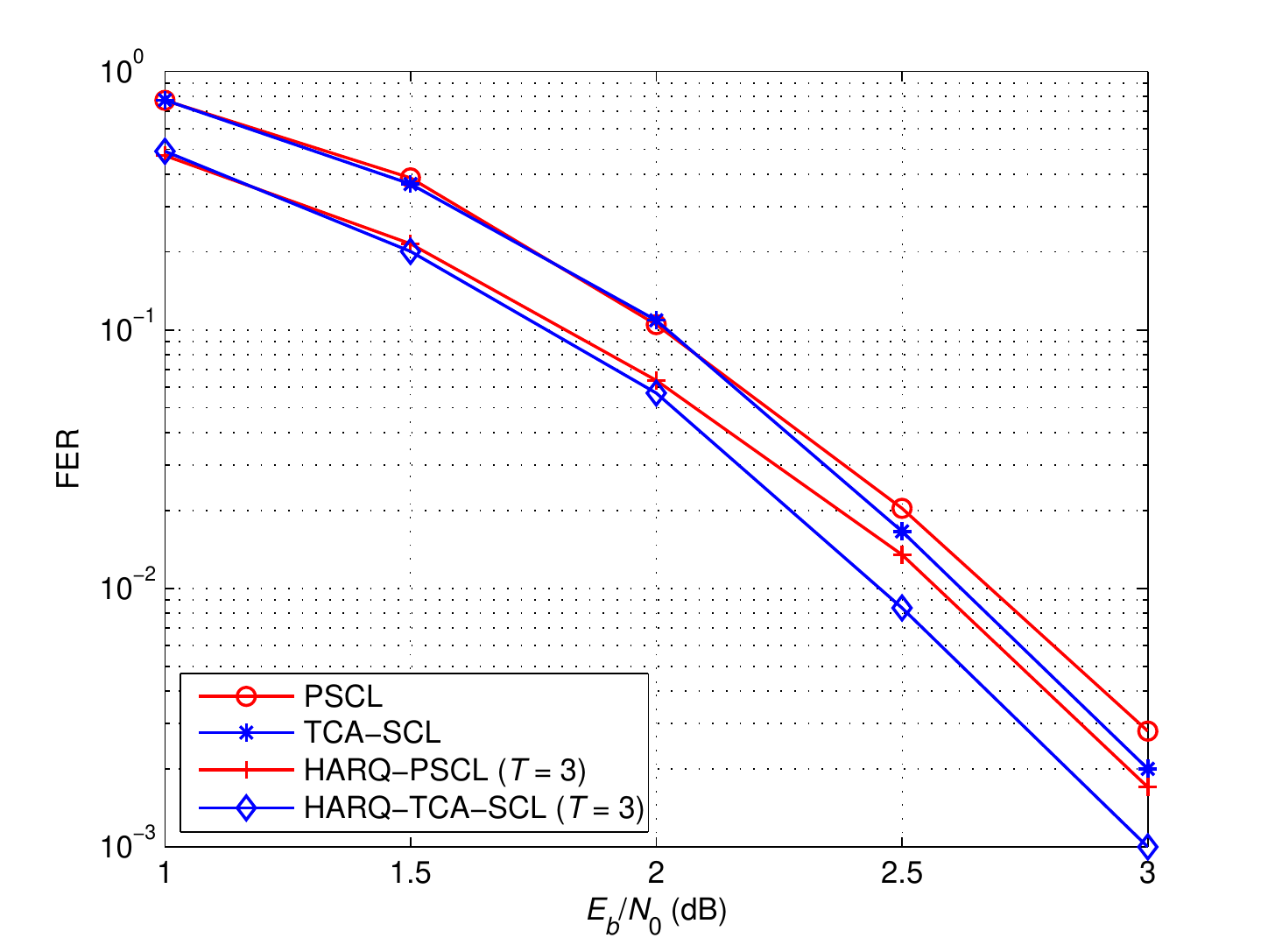}
}
\caption{FER comparison of (HARQ-)TCA-SCL and (HARQ-)PSCL schemes.}\label{fig:fer}
\end{figure}

According to Fig.~\ref{fig:fer}, compared with the PSCL scheme, the proposed TCA-SCL scheme has a $0.1$ dB performance gain when $\textrm{FER}=10^{-2}$ for both $(64,36)$ and $(1024,512)$ codes. The HARQ-TCA-SCL ($T = 3$) scheme introduces a $0.25$ dB and $0.13$ dB gain over the HARQ-PSCL scheme when frame error rate $\textrm{FER}=10^{-2}$ for $(64,36)$ and $(1024,512)$ codes, respectively.

\subsection{Complexity Analysis}\label{subsec:complexity}
Define the product of the actual decoding length and list size as the average list size. Since the average computational complexity is proportional to the average list size, here we analyze the average list sizes of TCA-SCL and HARQ-TCA-SCL decoders denoted by $\bar L_T$ and $\bar L_H$, respectively. Assume the total frame number is $F$, and the decoder ends at the $P_i$-th segment of the $i$-th frame. For the TCA-SCL decoder, $\bar L_T$ can be calculated as
\begin{equation}\label{eq:LT}
    {\bar L_T} = \frac{{L \times \sum\nolimits_{i = 1}^F {{P_i}} }}{{P \times F}}.
\end{equation}
Suppose the $i$-th frame is retransmitted $R_i$ times ($0 \le {R_i} \le T$). For the HARQ-TCA-SCL decoder, $\bar L_H$ is calculated as
\begin{equation}\label{eq:LH}
    {\bar L_H} = \frac{{L \times \sum\nolimits_{i = 1}^F {({P_i} + {R_i})} }}{{P \times F}}.
\end{equation}

For low SNR, thanks to the early termination $\bar L_T$ is small due to high error rate. On the other hand, a larger number of retransmissions leads to a higher $\bar L_H$ for HARQ-TCA-SCL decoder. As SNR increases, $\bar L_T$ and $\bar L_H$ converge to $L$: 1) TCA-SCL decoder is more likely to finish the decoding process, and 2) the retransmission time of HARQ-TCA-SCL decoder converges to $0$. It should be noted that, according to Eq.~(\ref{eq:LT}) and (\ref{eq:LH}) $0 \le {\bar L_H} - {\bar L_T} \le \frac{L}{P}T$.

\begin{figure} \centering
\subfigure[($N=64$, $K=36$, $m=8$, $P=2$)] { \label{fig:als64}
\includegraphics[width = 8 cm]{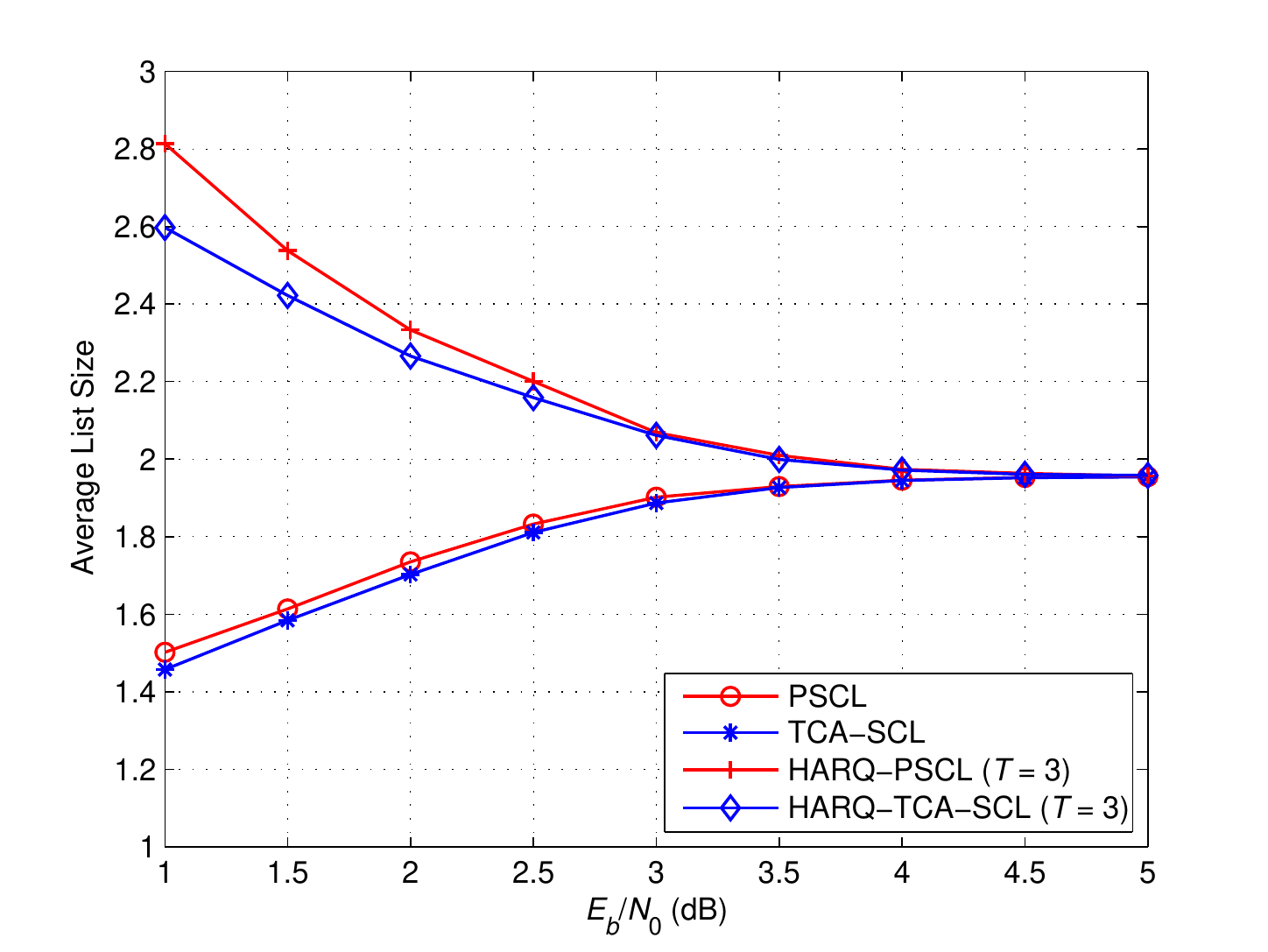}
}
\subfigure[($N=1024$, $K=512$, $m=32$, $P=4$)] { \label{fig:als1024}
\includegraphics[width = 8 cm]{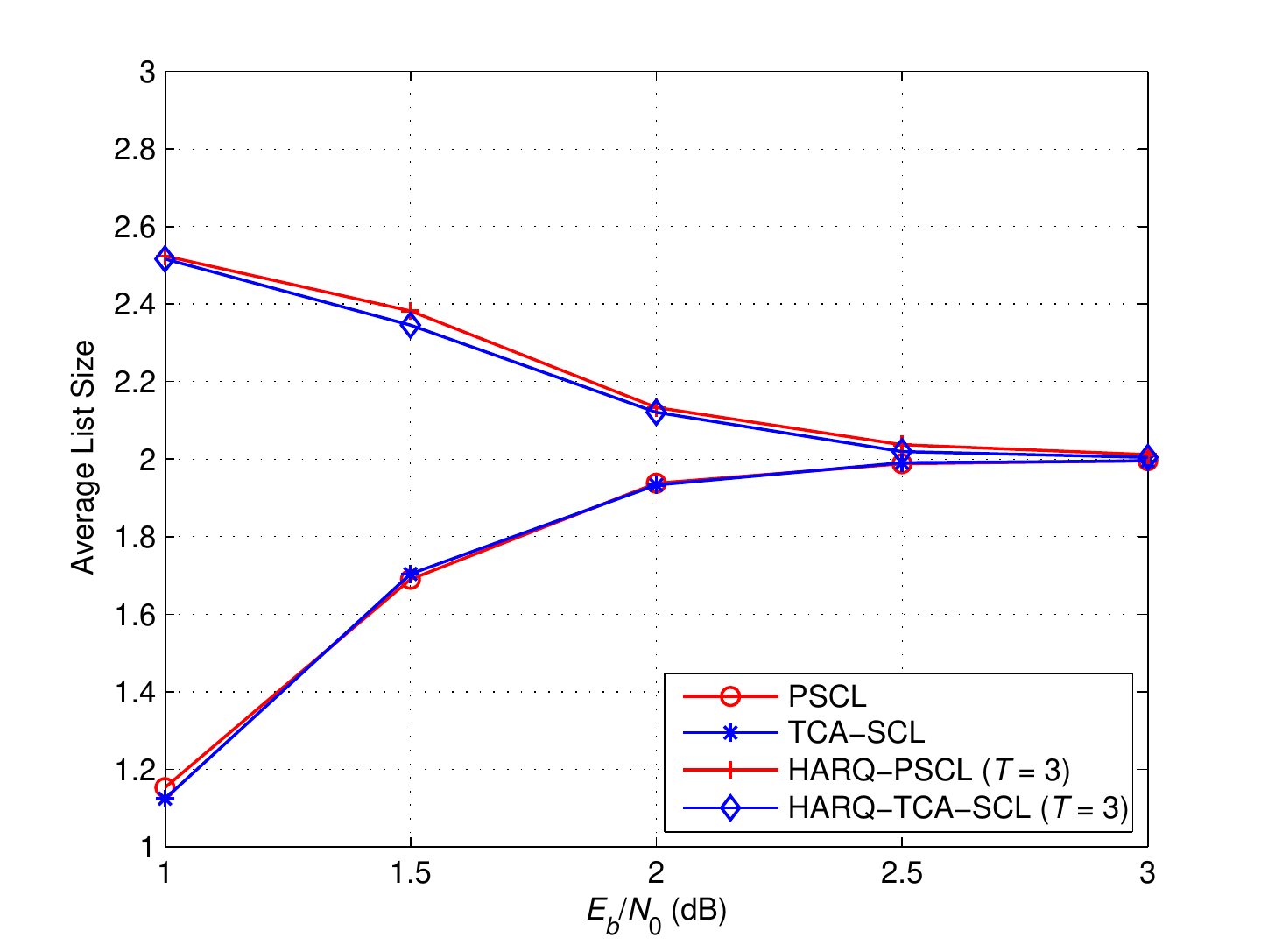}
}
\caption{Average list sizes of (HARQ-)TCA-SCL and (HARQ-)PSCL schemes.}\label{fig:als}
\end{figure}

Shown in Fig.~\ref{fig:als}, (HARQ-)TCA-SCL scheme has the same complexity as (HARQ-)PSCL scheme. The HARQ-TCA-SCL scheme has $50.3\%$ and $38.5\%$ higher complexity than the PSCL scheme at $\textrm{SNR}=1.5$ dB for $(64,36)$ and $(1024,512)$ codes, respectively. As SNR goes up, the complexity of HARQ-TCA-SCL scheme tends to be as same as the PSCL scheme asymptotically with better performance.

\section{Efficient TCA-SCL Decoder Architectures}\label{sec:arch}
To facilitate the application of the proposed TCA-SCL decoder, efficient architectures and FPGA implementations are proposed in this section and  are also given to demonstrate its merits. Since hardware consumption and decoding latency are two main concerns of SCL family decoder, the proposed architecture aims to achieve a good balance in between. The HARQ-TCA-SCL decoder can also be designed similarly.

\subsection{Hardware Consumption Analysis}\label{subsec:hardware}
\subsubsection{Full Module TCA-SCL Architecture}\label{subsubsec:arch1}
 \begin{figure*}[htb]
        \centering
            \includegraphics[width = 16cm]{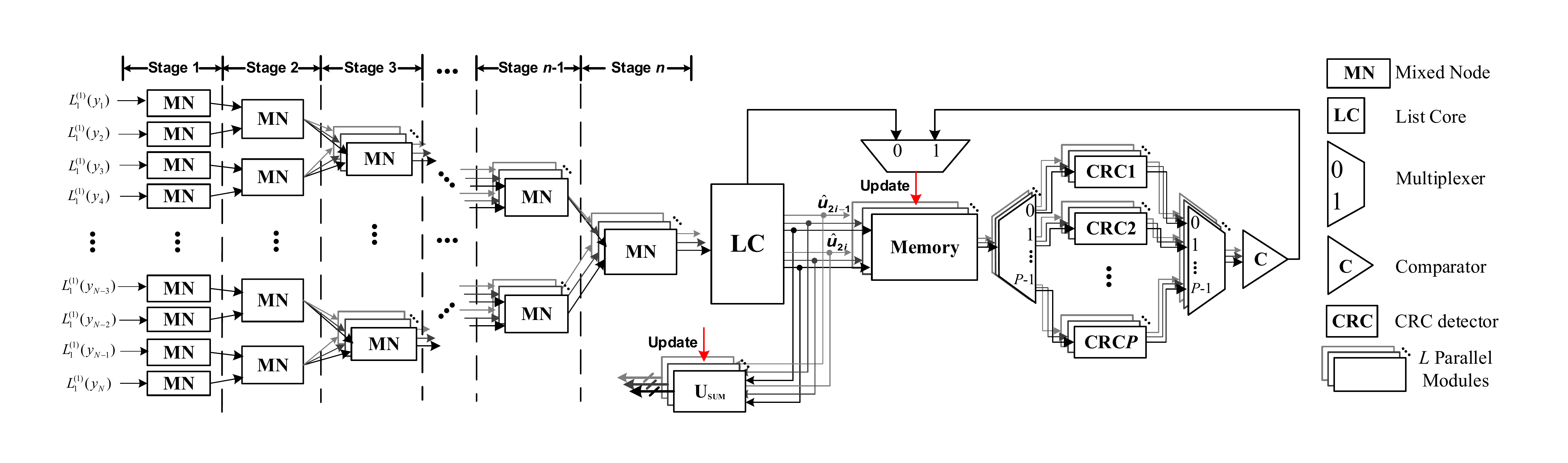}
        \caption{Architecture for full module TCA-SCL decoder.}\label{fig:adaf}
 \end{figure*}
 \begin{figure*}[htb]
        \centering
            \includegraphics[width = 16cm]{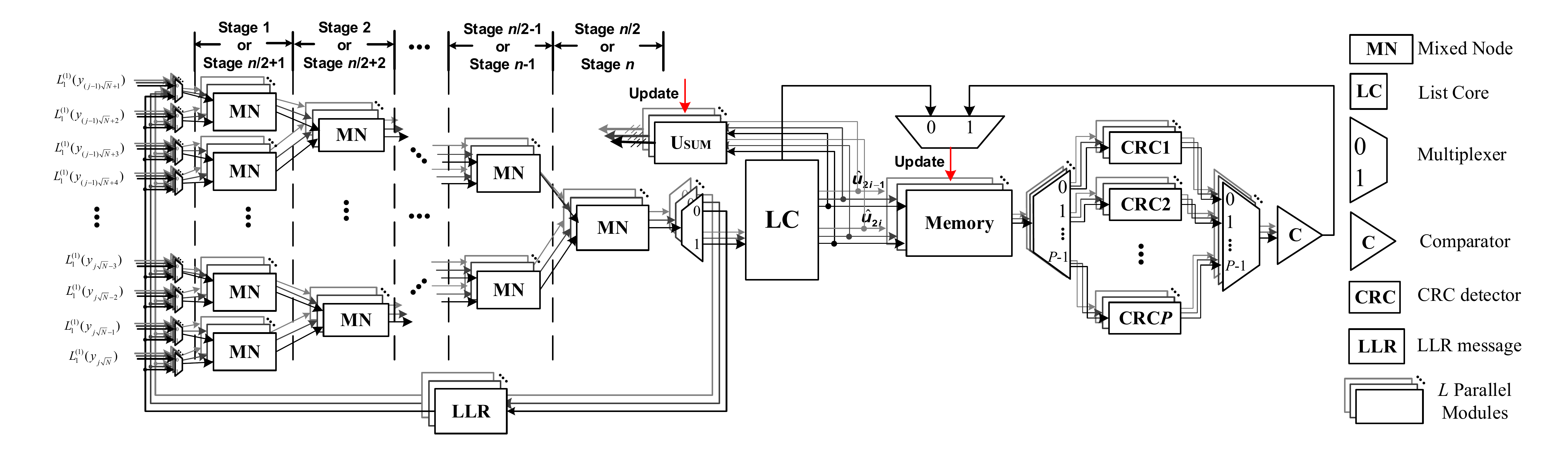}
        \caption{Architecture for folded TCA-SCL decoder ($n$ is even).}\label{fig:fsscl}
 \end{figure*}

In this subsection, a full module TCA-SCL architecture is proposed, which is mainly based on the conventional folded SC architecture proposed in~\cite{Zhang2011Reduced}. The architecture for full module TCA-SCL decoder is illustrated in Fig.~\ref{fig:adaf}. It divides all mixed node modules (MNs) into $n=\log_2 N$ stages, and each MN implements two types of calculations mentioned in Eq.~(\ref{eqn_dbl_x}). According to the conclusions in~\cite{Zhang2011Reduced}, for an $N$-bit SC decoder, $(N-1)$ MNs are required. For an $N$-bit CA-SCL decoder, $L(N-1)$ MNs are employed.

\begin{Thm}\label{thm:full:mn}
  For one $P$-segmented TCA-SCL decoder with list $L$, the total number of MNs is
  \begin{equation}
  \textstyle
    \mathrm{MN}_\mathrm{total}=N-L+(L-1)\frac{N}{P}.
  \end{equation}
\end{Thm}

\begin{proof}\label{proof:full:mn}
  For the given decoder, its MNs can be categorized into two parts. The first part includes Stages $1$ to $\log_2 P$. The second part includes Stages $(\log_2 P+1)$ to $n$. It should be noted that since $N$ is power of $2$, $\log_2 P$ is always an integer.

  Since each segment outputs only one candidate, the first part obeys SC decoding rule, and list size $L$ is not necessary. The number of MNs is
  \begin{equation}\label{eq:full:mn:1}
  \textstyle
    \textrm{MN}_1=\sum\nolimits_{i=1}^{\log_2 P} N/{2^{i-1}}=N-\frac{N}{P}.
  \end{equation}

  The second part obeys CA-SCL decoding rule without considering the fine-gain scheduling. The number of MNs is
  \begin{equation}\label{eq:full:mn:2}
  \textstyle
    \textrm{MN}_2=L\sum\nolimits_{i=\log_2 P+1}^{n} N/{2^{i-1}}=L(\frac{N}{P}-1).
  \end{equation}
\end{proof}

\begin{table*}[ht]
    \tabcolsep 1mm
    \renewcommand{\arraystretch}{1.2}
    \footnotesize
    \caption{Implementation Analysis for Different Schemes}
    \begin{center}
    \begin{tabular}{c || c | c | c | c | c | c | c }
    \Xhline{1.0pt}
    \multirow{2}{*}{\textbf{Schemes}} &\multicolumn{2}{c|}{\multirow{2}{*}{Mixed node \#}}& \multicolumn{4}{c|}{Memory (bit)} & \multirow{2}{*}{Latency (clock cycles)}\\
    \cline{4-7}&\multicolumn{2}{c|}{} & \multicolumn{2}{c|}{LLRs} & \multicolumn{2}{c|}{Outputs} & \\
    \hline
     \textbf{CA-SCL}                 & $(N-1)L$                       &{\color{red}{$2046$}}&$(N-1)Lq$& {\color{red}{$2046q$}} &$(K+m)L$ &{\color{red}{$1088$}}&$T_\mathrm{CA}$  \\
     \textbf{TCA-SCL (SF)} & $N-L+\frac{(L-1)N}{P}$ &{\color{red}{$1278$}}&$(N-L+\frac{(L-1)N}{P})q$& {\color{red}{$1278q$}} &$(K+m)L$ &{\color{red}{$1088$}}&$T_\mathrm{CA}+(T_1+...+T_{(P-1)})L$  \\
     \textbf{TCA-SCL (DF)}  & $N-L+\frac{(L-1)N}{P}$ &{\color{red}{$1278$}}&$(N-L+\frac{(L-1)N}{P})q$& {\color{red}{$1278q$}} &$2(K+m)L$&{\color{red}{$2176$}}&$T_\mathrm{CA}+P\log_2 P -2P +2$ \\
     \textbf{FTCA-SCL (SF)}& $(\sqrt{N}-1)L$                &{\color{red}{$62$}}&$(N-L+\frac{(L-1)N}{P})q$&{\color{red}{$1278q$}}&$(K+m)L$ &{\color{red}{$1088$}}&$T_\mathrm{CA}+F+(T_1+...+T_{(P-1)})L$  \\
     \textbf{FTCA-SCL (DF)}& $(\sqrt{N}-1)L$                &{\color{red}{$62$}}&$(N-L+\frac{(L-1)N}{P})q$&{\color{red}{$1278q$}}&$2(K+m)L$&{\color{red}{$2176$}}&$T_\mathrm{CA}+F+P\log_2 P -2P +2$  \\
    \Xhline{1.0pt}
    \end{tabular}\label{tab:consult}
    \end{center}
    \end{table*}

Since the memory block corresponds to MNs, the memory complexity is as follows

\begin{Cor}\label{cor:full:mn}
  Assume the quantization length for the LLR message is $q$, the memory bits required are
  \begin{equation}
  \textstyle
    \mathrm{mem}_\mathrm{total}=q \times \mathrm{MN}_\mathrm{total}=q\left(N-L+(L-1)\frac{N}{P}\right).
  \end{equation}
\end{Cor}

The list core (LC) module in Fig.~\ref{fig:adaf} mainly implements the sorting operation. In order to reduce both the sorting latency and complexity, the efficient distributed sorting (DS) proposed in~\cite{2016Liang} is employed here.

\subsubsection{Folded Module TCA-SCL Architecture}\label{subsubsec:arch2}
Thanks to the early termination scheme, the proposed full module architecture for TCA-SCL decoding is memory efficient compared to conventional CA-SCL decoding. However, the hardware utilization ratio (HUR) of MNs is very low. Borrowing the fine-folding idea proposed in~\cite{2016Liang,2016Liang2}, this paper then proposes the folded module TCA-SCL architecture for higher HUR. We set up a sub-decoder with $(2^{\lceil n/2 \rceil}-1)L$ MNs for Stage $1$ to $\lceil n/2 \rceil$. Stage $({\lceil n/2 \rceil}+1)$ to $n$ can also be implemented by this sub-decoder in a time-multiplexing manner. Fig.~\ref{fig:fsscl} gives an example of a even $n$, Stage $1$ and Stage $n/2+1$, Stage $2$ and Stage $n/2+2$, ..., Stage $n/2$ and Stage $n$ are time-multiplexing. If $n$ is odd, Stage $1$ and Stage $(n+1)/2+1$, Stage $2$ and Stage $(n+1)/2+2$, ..., Stage $(n+1)/2-1$ and Stage $n$ are time-multiplexing, and Stage $(n+1)/2$ uses the last stage alone. Parameter $j$ in Fig.~\ref{fig:fsscl} denotes the current folding order. However, the characteristics in Section~\ref{subsubsec:arch1} which helps to reduce the complexity of the first $\log_2 P$ stages could not be employed here, because folding technique is based on uniform hardware. The complexity is
\begin{Thm}
  For one folded module TCA-SCL decoder with list $L$, the total number of MNs is
  \begin{equation}
    {\rm{MN}}_{\rm{total}} = (2^{\lceil n/2 \rceil}-1)L.
  \end{equation}
\end{Thm}

\begin{proof}
  When implementing Stage $1$ to $\lceil n/2 \rceil$, all the input and output multiplexers choose mode `$0$'. $2^{\lceil n/2 \rceil + 1}$ executions are required to output $2^{\lceil n/2 \rceil + 1} L$ LLRs for Stage $\lceil n/2 \rceil$. For $P$-segmented decoder, if $\log_2 P \geq \lceil n/2 \rceil$, $(2^{\lceil n/2 \rceil}-1)(L - 1)$ MNs are idle during this decoding stage. Otherwise, according to Eq.~(\ref{eq:full:mn:1}), $(2^{\lceil n/2 \rceil}  - \frac{2^{\lceil n/2 \rceil}}{P})(L - 1)$ MNs are idle. Therefore, $(2^{\lceil n/2 \rceil}-1)L$ MNs are sufficient.

  When implementing Stage $({\lceil n/2 \rceil}+1)$ to $n$, all the input and output multiplexers choose mode `$1$'. Since $2^{\lceil n/2 \rceil + 1} L$ LLRs become the input of the sub-decoder, no MN is idle during this stage. Therefore, the total number of MNs is $(2^{\lceil n/2 \rceil}-1)L$.
\end{proof}

\begin{Thm}\label{thm:latency}
  Assume the quantization length for the LLR message is $q$, the memory bits required by the folded module TCA-SCL architecture is\\
  $\textstyle
    {\rm{me}}{{\rm{m}}_{{\rm{total}}}} = q\left(N-L+(L-1)\frac{N}{P}\right)$.
\end{Thm}
\begin{proof}
  The folded design only reduces the complexity of MNs. However, the memory complexity stays the same as the full module TCA-SCL architecture.
\end{proof}
Table~\ref{tab:implementation} gives FPGA results in accordance with Theorem~\ref{thm:latency}.

\subsection{Timing Analysis}\label{subsec:ta}
\begin{figure*} \centering
\subfigure[TCA-SCL polar decoder with SF scheme.] { \label{fig:la1}
\includegraphics[width = 0.85\textwidth]{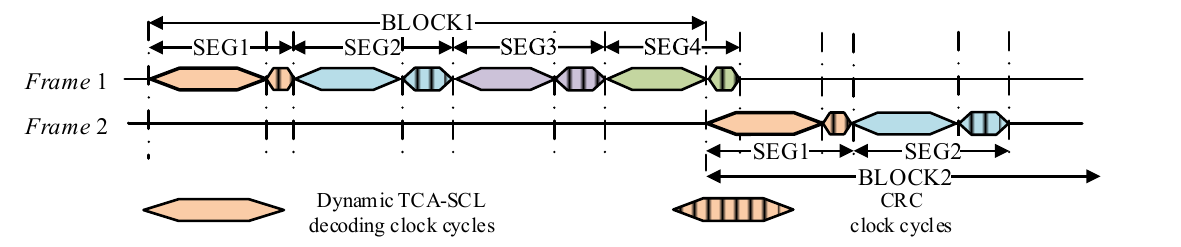}
}
\subfigure[TCA-SCL polar decoder with DF scheme.] { \label{fig:la}
\includegraphics[width = 0.85\textwidth]{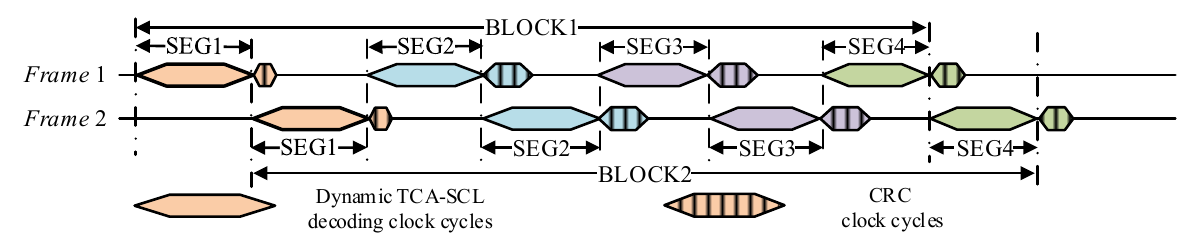}
}
\caption{Timing analysis for TCA-SCL polar decoder with SF and DF schemes.}\label{fig:lat}
\end{figure*}

\subsubsection{Single Frame Scheme}
As Fig.~\ref{fig:adaf} shown, the decoding process for TCA-SCL has the following steps: 1) In Segment $j$, MNs complete the main decoding in Eq.~(\ref{eqn_dbl_x}). The $2L$ LLRs correspond to $\hat{u}_i$ for each path. 2) $2L$ LLRs are input to the LC module. DS method~\cite{2016Liang} is employed to select the best $L$ paths. 3) The memory is updated and partial sum vector $\hat{\mathbf{u}}_{sum}$ is calculated for $\hat{u}_{i+1}$. 4) We repeat the above steps to get the $L$ paths for $\hat{u}^{(\frac{N}{P}j-1)}_1$. Then, $\hat{u}_{\frac{N}{P}j}$ is directly chosen as `$0$' or `$1$' for each path without decoding. After that, we input information bits in $\hat{u}^{\frac{N}{P}j}_{\frac{N}{P}(j-1)}$ for $2L$ paths to CRC$_j$ to pick up the only path for Segment $j+1$. CRC is implemented with linear feedback shift register (LFSR)~\cite{Lin2004Error}, and determines the coefficient of \textsc{xor}. Shown in Fig.~\ref{fig:adaf}, $P$ CRC modules are employed. It should be noted that here CRC$_j$ takes care of $2L$ paths in serial manner. Admittedly, designers can process $2L$ with parallel CRCs. Considering the simple CRC and its short processing time, serial manner is employed here.

The scheduling of this single frame (SF) scheme is shown in Fig.~\ref{fig:la1}. The latency of SF TCA-SCL decoder is
\begin{Thm}
  Assume the latency of CA-SCL is $T_{\rm{CA}}$ clock cycles. The latency for CRC$_i$ is $T_i$. For one SF $P$-segmented TCA-SCL decoder with list $L$, the decoding latency is
  \begin{equation}
    T_\mathrm{SF}=T_\mathrm{CA}+2L(T_1+...+T_{(P-1)}).
  \end{equation}
\end{Thm}

\begin{proof}
  After checking all $2L$ paths of Segment $i$, the decoder selects one path and begins to decode Segment $(i+1)$. In SF scheme, segmented CRC scheme increases latency for serial checking of Segment $i$
  \begin{equation}
    T_{\mathrm{crc}_i}=2LT_i.
  \end{equation}

  In SF scheme, checking Segment $P$ of Frame $1$ and decoding Segment $1$ of Frame $2$ can be done at the same time. Since the checking time is shorter than decoding time, the latency increase is
  \begin{equation}
    T_\mathrm{inc}=2L\sum\nolimits_{i=1}^{P-1} T_i.
  \end{equation}

  Now the proof is immediate.
\end{proof}

Folded module TCA-SCL decoder can also work in the proposed SF scheme.
\begin{Cor}
  Assume the folding technique introduces $F$ extra clock cycles per frame, the latency of SF folded module TCA-SCL decoder is
  \begin{equation}
    T_\mathrm{SF}=T_{\rm{CA}}+F+2L(T_1+...+T_{(P-1)}).
  \end{equation}
\end{Cor}

\subsubsection{Double Frame Scheme}
SF decoding introduces $2L(T_1+...+T_{(P-1)})$ extra clock cycles per frame. During CRC detection, all MNs are idle and HUR is therefore low. To this end, the double frame (DF) scheme is proposed.

The main idea of DF is shown in Fig.~\ref{fig:la}. Two frames are decoded simultaneously in an interleaved manner: when Frame $1$ checks (decodes) its Segment $i$, Frame $2$ decodes (checks) its Segment $i$ ($i-1$). Since both frames share the same architecture, every time a new segment is decoded, all LLRs in memory belong to the other frame. If we keep the decoding latency of each frame the same as CA-SCL decoder, Stage $1$ to $(\log_2 {P}-1)$ need an extra memory block of $q(\frac{N}{2}+\frac{N}{4}+...+\frac{2N}{P})$ bits to save LLRs, which is not appreciated by hardware design.

If no extra memory is available, each new segment begins its decoding with Stage $1$. In this way, DF scheme still requires a memory block of\\ $q\left(N-L+(L-1)\frac{N}{P}\right)$ bits with slightly increased latency. For DF full module TCA-SCL decoder:

\begin{Thm}
  For one DF $P$-segmented TCA-SCL decoder with list $L$, the decoding latency is
  \begin{equation}
    T_\mathrm{DF}=T_\mathrm{CA}+P\log_2 P -2P +2.
  \end{equation}
\end{Thm}

\begin{proof}
 For the interleaved manner in Fig.~\ref{fig:la}, the latency of each segment is
   \begin{equation}
    T_{\mathrm{seg}_{i}}=\max\{T_{\mathrm{dec}_{i}},T_{\mathrm{crc}_{i}}\},
  \end{equation}
  where $T_{\textrm{dec}_{i}}$ denotes the SCL decoding latency for Segment $i$, which includes SC decoding and DS. According to~\cite{2016Liang}, the DS latency for Segment $i$ is approximately $2LT_i$, therefore
  \begin{equation}
   T_{\textrm{dec}_{i}}>2LT_i.
  \end{equation}
 Since $T_{\textrm{crc}_{i}}=2LT_i$
     \begin{equation}
    T_\textrm{DF}=\sum\nolimits_{i=1}^{P}T_{\textrm{seg}_{i}}=\sum\nolimits_{i=1}^{P}T_{\textrm{dec}_{i}}.
  \end{equation}
  It is believed that there are $2^i$ segments, which could calculate from Stage $(i+1)$, now calculates from Stage $1$ and introduce latency of $i \cdot 2^i$. Therefore, the decoding increased latency is
  \begin{equation}
    T_\textrm{inc}=\sum\nolimits_{i=1}^{\log_2{P}-1} i \cdot 2^i=P\log_2{P}-2P+2.
  \end{equation}

  Now the proof is immediate.
\end{proof}

Folded module TCA-SCL decoder can also work in the proposed DF scheme with the following latency.

\begin{Cor}
  Assume the folding technique introduces $F$ extra clock cycles per frame, the latency of DF folded module TCA-SCL decoder is
  \begin{equation}
    T_\mathrm{DF}=T_\mathrm{CA}+F+P\log_2{P} -2P +2.
  \end{equation}
\end{Cor}

Table~\ref{tab:consult} shows comparison between five different schemes: CA-SCL decoder, SF (DF) full module TCA-SCL decoders, and SF (DF) folded module TCA-SCL decoders. According to Section~\ref{subsec:vis}, CRC allocation is \\
$(|C_1|, |C_2|, |C_3|, |C_4|) = (3, 10, 11, 8)$. Data in red show the example of $N=1024$, $K=512$, $P=4$, and $L=2$.

\subsection{FPGA Implementation Results}\label{subsec:imp}
To better demonstrate the advantages of the proposed TCA-SCL decoders, FPGA implementations based on Altera Stratix V are given as well. To be in accordance with Table~\ref{tab:implementation}, five decoders have been implemented. The same parameters as the aforementioned example are employed here: $N=1024$, $K=512$, $m=32$, $P=4$, and $L=2$. All the five decoders employ the same LLR quantization scheme of $1$ sign bit, $6$ integer bits, and $1$ decimal bit. In Fig.~\ref{fig:q}, the FER performance comparison of floating SC and quantized-SC with $q=8$ bits indicates the validity of the quantized scheme.
    \begin{figure}[htb]
        \centering
            \includegraphics[width = 8cm]{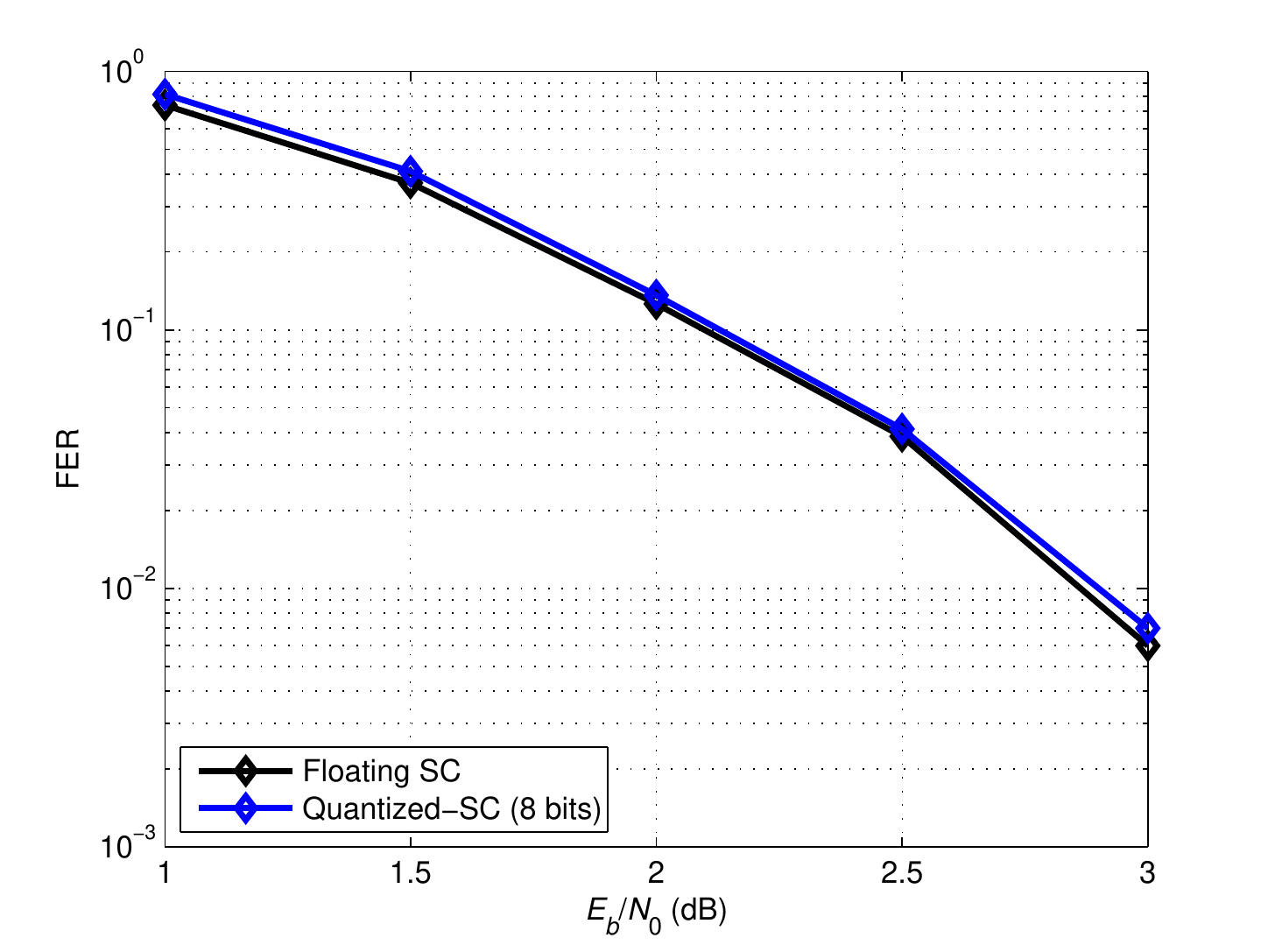}
        \caption{Performance comparison regarding quantization ($N = 1024$, $K = 512$).}\label{fig:q}
    \end{figure}

The implementation results are compared in terms of adaptive logic modules (ALMs), registers, and memory bits. It is shown that, compared to the CA-SCL decoder, TCA-SCL (SF or DF) decoder can achieve $18.8\%$ or $15.0\%$ ALM reduction. For further ALM reduction, with the help of folding technique, FTCA-SCL (SF or DF) decoder consumes $40.11\%$ or $42.9\%$ ALMs compared to TCA-SCL (SF or DF) with slightly increased latency, as analyzed in~\cite{2016Liang2}. It is also observed that the ALMs' reduction is not that much as the reduction of MNs listed in Table~\ref{tab:consult}. This is because Table~\ref{tab:consult} does not consider the comparison part, which introduces major part of ALMs consumption and stays the same between different architectures.

For implementation convenience, here memory has been employed by both folded decoders. Therefore, we consider the sum of registers and memory bits as the total memory consumption. It is observed TCA-SCL (SF or DF) decoder requires $77.03\%$ or $82.1\%$ memory compared to the CA-SCL decoder. Also, the introduction of folding technique does not affect the memory cost, which has been indicated by Theorem~\ref{thm:latency}. Comparing FTCA-SCL (DF) decoder and FTCA-SCL (SF) decoder, when DF scheme is employed, the latency can be reduced $15.90\%$ at the cost of $11.99\%$ increased ALMs.

For the latency issue, since the critical paths of all designs are determined by the critical path of the same SC decoding kernel, we believe it is safe to compare in term of clock number. It is shown that the segmented CRC decoders will introduce more latency due to more serial CRC operations. Second, the DF scheme is more time efficient. Third, the folded versions come at the cost of higher latency.

In general, the proposed four architecture of DC-SCL decoding can reduce the hardware consumption compared to CA-SCL decoder. Designers can choose the suitable one according to different application requirements.
\begin{table}[ht]
    \tabcolsep 1mm
    \renewcommand{\arraystretch}{1.2}
    \footnotesize
    \caption{FPGA Implementation Results for Different Schemes}
    \begin{center}
    \begin{tabular}{c || c | c | c | c}
    \Xhline{1.0pt}
    \textbf{Schemes} & ALMs & Registers & Memory & Latency\\
    \hline
     \textbf{CA-SCL}               &   $102,847$    &  $20,064$   &  $0$  &   $2655$   \\
     \textbf{TCA-SCL (SF) }        &    $83,529$    &   $15,456$  &  $0$  &   $3253$   \\
     \textbf{TCA-SCL (DF) }        &    $87,454$    &  $16,480$   &  $0$  &   $2657$   \\
     \textbf{FTCA-SCL (SF)}        &    $33,502$    &   $5,515$  &   $11,264$  &   $3749$   \\
     \textbf{FTCA-SCL (DF)}        &     $37,518$   &   $6,558$  &   $11,264$  &   $3153$   \\
    \Xhline{1.0pt}
    \end{tabular}\label{tab:implementation}
    \end{center}
    \end{table}

\section{Conclusions}\label{sec:CON}
In this paper, a segmented SCL polar decoding with tailored CRC is proposed. Method on how to choose the proper CRC for a given segment is proposed with help of concepts of virtual transform and virtual length. Numerical results have shown that the proposed TCA-SCL decoder can achieve better performance and lower complexity than conventional CA-SCL decoder. Thanks to the more reasonable CRC partition scheme, the TCA-SCL decoder can also outperform the PSCL decoder. For further performance improvement, HARQ-TCA-SCL scheme is proposed at the cost of increased complexity. Efficient architectures and FPGA implementations are also proposed for a good balance between hardware consumption and decoding latency.

\bibliographystyle{IEEEtran}
\bibliography{IEEEabrv,mybib}
\end{document}